\documentclass[twocolumn]{autart}
\usepackage{epsfig}
\usepackage{epstopdf}
\usepackage{graphicx}
\usepackage[mathscr]{eucal}
\usepackage{amsmath,amsfonts,amssymb,amsbsy,amsopn,amstext}
\usepackage{latexsym}
\usepackage{color,multicol,makeidx}
\usepackage{graphicx}
\usepackage{bm}
\usepackage{mathrsfs}
\usepackage{float}
\usepackage{ccaption}
\usepackage{cite}
\usepackage{natbib}
\usepackage{paralist}
\usepackage[caption=false,font=normalsize,labelfont=sf,textfont=sf]{subfig}
\usepackage{enumerate}
\usepackage{tabu}
\usepackage{listings}
\def\dref#1{(\ref{#1})}
\def\eq{\displaystyle\stackrel\triangle=}
\def\xra{\xrightarrow}

\newcommand{\diag}{\ensuremath{\mathrm{diag}}}
\DeclareMathOperator*{\argmin}{arg\,min}

\newenvironment{proof}[1][Proof]{\noindent\textbf{#1.} }{\hfill \rule{0.5em}{0.5em}}

\begin{document}

\begin{frontmatter}
\runtitle{}

\title{On Asymptotic Properties of Hyperparameter Estimators
    for Kernel-based Regularization Methods
\thanksref{footnoteinfo}}

\thanks[footnoteinfo]{This work was supported by the Thousand Youth Talents Plan funded by
the central government of China, the Shenzhen Project Ji-20160207,
the President's grant under contract No. PF. 01.000249 and the
Start-up grant under contract No. 2014.0003.23 funded by the Chinese
University of Hong Kong, Shenzhen, as well as by a research grant
for junior researchers under contract No. 2014-5894, funded by
Swedish Research Council. }

\author[LIN]{Biqiang Mu}\ead{biqiang.mu@liu.se},
\author[CUHK]{Tianshi Chen}\ead{tschen@cuhk.edu.cn},
and \author[LIN]{Lennart Ljung}\ead{ljung@isy.liu.se}

\address[LIN]
{Division of Automatic Control, Department of Electrical Engineering, Link\"oping University, Link\"oping SE-58183, Sweden}
\address[CUHK]{School of Science and Engineering, The Chinese University of Hong Kong, Shenzhen, China}


\begin{abstract}
    The kernel-based regularization method has two core issues: kernel
    design and hyperparameter estimation. In this paper, we focus on the second issue and study the
    properties of several hyperparameter estimators including the
    empirical Bayes (EB) estimator, two Stein's unbiased risk
    estimators (SURE) and their corresponding Oracle counterparts, with an emphasis on the asymptotic properties of these hyperparameter estimators. To this goal, we first derive and
    then rewrite the first order optimality conditions of these hyperparameter estimators, leading to several insights on these
    hyperparameter estimators. Then we show that as the number of data
    goes to infinity, the two SUREs converge to the best hyperparameter
    minimizing the corresponding mean square error, respectively, while
    the more widely used EB estimator converges to another best
    hyperparameter minimizing the expectation of the EB estimation criterion. This indicates that
    the two SUREs are asymptotically optimal but the EB estimator is
    not. Surprisingly, the convergence rate of two SUREs is slower than
    that of the EB estimator, and moreover, unlike the two SUREs, the EB
    estimator is independent of the convergence rate of $\Phi^T\Phi/N$ to
    its limit, where $\Phi$ is the regression matrix and $N$ is the
    number of data. A Monte Carlo simulation is provided to demonstrate
    the theoretical results.
\end{abstract}

\end{frontmatter}

\section{Introduction}

The kernel-based regularization methods (KRM) from machine learning
and statistics were first introduced to the system identification
community in \cite{Pillonetto2010} and then further developed in
\cite{Pillonetto2011,Chen2012,Chen2014p}. These methods attract
increasing attention in the community and have become a complement
to the classical maximum likelihood/prediction error methods
(ML/PEM) \citep{Pillonetto2015,Chen2012,Ljung2015}. In particular,
KRM may have better average accuracy and robustness than ML/PEM when
the data is short and/or has low signal-to-noise ratio (SNR).

There are two core issues for KRM: kernel design and hyperparameter
estimation. The former is regarding how to parameterize the kernel matrix
with a parameter vector, called hyperparameter, to embed the prior
knowledge of the system to be identified, and the latter is
regarding how to estimate the hyperparameter based on the data such
that the resulting model estimator achieves a good bias-variance
tradeoff or equivalently, suitably balances the adherence to the
data and the model complexity.

The kernel design plays a similar role as the model structure design
for ML/PEM and determines the underlying model structure for KRM. In
the past few years, many efforts have been spent on this issue and
several kernels have been invented to embed various types of prior
knowledge, e.g.,
\cite{Pillonetto2010,Pillonetto2011,Chen2012,Chen2014p,Dinuzzo2015,Chen2016,Carli2017,Marconato2016,Zorzi2017,Pillonetto2016}.
In particular, two systematic kernel design methods (one is from a
machine learning perspective and the other one is from a system
theory perspective) were developed in \cite{Chen2017} by embedding
the corresponding type of prior knowledge.

The hyperparameter estimation plays a similar role as the model
order selection in ML/PEM and its essence is to determine a suitable
model complexity based on the data. As mentioned in the survey of
KRM \cite{Pillonetto2014}, many methods can be used for
hyperparameter estimation, such as the cross-validation (CV),
empirical Bayes (EB), $C_p$ statistics and Stein's unbiased risk
estimator (SURE) and etc. In contrast with the numerous results on
kernel design, there are however few results on hyperparameter
estimation except
\cite{Chen2014p,Aravkin2012,Aravkin2012c2,Aravkin2014,Pillonetto2015}.
In \cite{Aravkin2012,Aravkin2012c2,Aravkin2014}, two types of
diagonal kernel matrices are considered. When $\Phi^T\Phi/N$ is an
identity matrix, where $\Phi$ is the regression matrix and $N$ is
the number of data, the optimal hyperparameter estimate of the EB
estimator has explicit form and is shown to be consistent in terms
of the mean square error (MSE). When $\Phi^T\Phi/N$ is not an
identity matrix, the EB estimator is shown to asymptotically
minimize a weighted MSE. In \cite{Chen2014p}, the EB with linear
multiple kernel is shown to be a difference of convex programming
problem and moreover, the optimal hyperparameter estimate is sparse.
In \cite{Pillonetto2015}, an unbiased estimator of MSE was
introduced and used as a measure to evaluate the performance of the
EB estimator and two SUREs: one for impulse response reconstruction
and the other one for output prediction, and the robustness issue by
introducing the so-called excess degree of freedom was considered.

In this paper, we study the properties of the EB estimator and two
SUREs in \cite{Pillonetto2015} with an emphasis on the asymptotic
properties of these hyperparameter estimators. In particular, we are
interested in the following questions: When the number of data goes
to infinity, \begin{itemize}
\item[1)] what will be the best kernel matrix, or equivalently, the best value of the
hyperparameter?

\item[2)] which estimator (method) shall be chosen such
that the hyperparameter estimate tends to this best value in the
given sense?

\item[3)] what will be the convergence rate
of that the hyperparameter estimate tends to this best value? and
what factors does this rate depend on?

%
%

\end{itemize}

In order to answer these questions, we employ the regularized least
squares method for FIR model estimation in \cite{Chen2012}. As a
motivation, we first show that the regularized least squares
estimate can have smaller MSE than the least squares estimate for
any data length, if the kernel matrix is chosen carefully. We then
derive the first order optimality conditions of these hyperparameter
estimators and their corresponding Oracle counterparts (relying on
the true impulse response, see Section \ref{sec:hyper_estimator} for
details). These first order optimality conditions are then rewritten
in a way to better expose their relations, leading to several
 insights on these hyperparameter estimators. For instance,
one insight is that for the Oracle estimators, for any data length,
and without structure constraints on the kernel matrix, the optimal
kernel matrices are same as the one in \cite{Chen2012} and equal to
the outer product of the vector of the true impulse response and its
transpose. Moreover, explicit solutions of the optimal
hyperparameter estimate for two special cases are derived
accordingly. Then we turn to the asymptotic analysis of these
hyperparameter estimators. Regardless of the parameterization of the
kernel matrix, we first show that the two SUREs actually converge to
the best hyperparameter minimizing the corresponding MSE,
respectively, as the number of data goes to infinity, while the more
widely used EB estimator converges to the best hyperparameter
minimizing the expectation of the EB estimation criterion. In
general, these best hyperparameters are different from each other
except for some special cases. This means that the two SUREs are
asymptotically optimal but the EB estimator is not. We then show
that the convergence rate of two SUREs is slower that of the EB
estimator, and moreover, unlike the two SUREs, the EB estimator is
independent of the convergence rate of $\Phi^T\Phi/N$ to its limit.

The remaining parts of the paper is organized as follows. In Section
\ref{sec2}, we recap the regularized least squares method for FIR
model estimation and introduce two types of MSE. In Section
\ref{sec3}, we introduce a couple of widely used parameterizations
of kernel matrix and six hyperparameter estimators, including the EB
estimator, two SUREs, and their corresponding Oracle counterparts.
In Section \ref{sec4}, we derive the first order optimal conditions
of these hyperparameter estimators and put them in a form that
clearly shows their relation, leading to several insights. In
Section \ref{sec5}, we give the asymptotic analysis of these
hyperparameter estimators, including the asymptotic convergence and
the corresponding convergence rate. In Section \ref{sec6}, we
illustrate our theoretical results with a Monte Carlo simulation.
Finally, we conclude this paper in Section \ref{sec7}. All proofs of
the theoretical results (propositions, corollaries and theorems) are
postponed to the Appendix.

\section{Regularized Least Squares Approach for FIR Model Estimation}
\label{sec2}
Consider a single-input single-output linear discrete-time invariant, stable and causal system
\begin{align}
y(t) = G_0(q^{-1})u(t) + v(t),\ t=1,\dots, N\label{sys}
\end{align}
where $t$ is the time index, $y(t),u(t),v(t)$ are the output, input and disturbance of the system at time  $t$, respectively,   $G_0(q^{-1})$ is the transfer function of the system and $q^{-1}$ is the backshift operator: $q^{-1} u(t)=u(t-1)$. Assume that the input $u(t)$ is known (deterministic) and the input-output data are collected at time instants $t=1,\cdots,N$, and moreover, the disturbance $v(t)$ is a zero mean white noise with variance $\sigma^2>0$. The problem is to estimate a model for $G_0(q^{-1})$  as well as possible based on the the available data $\{u(t-1),y(t)\}_{t=1}^N$.

The transfer function $G_0(q^{-1})$ can be written as
\begin{align}\label{eq:truesys}G_0(q^{-1}) =\sum_{k=1}^\infty g_k^0 q^{-k},\end{align}where $g_k^0,k=1,\cdots,\infty$ form the impulse response of the system. Since the impulse response of a stable linear system decays exponentially, it is possible to truncate the infinite impulse response at a sufficiently high order, leading to the finite impulse response (FIR) model:
\begin{align}\label{eq:FIR}
G(q^{-1}) =\sum_{k=1}^ng_k q^{-k},~~\theta=[g_1,\cdots,g_n]^T \in \mathbb{R}^n.
\end{align}
With the FIR model (\ref{eq:FIR}), system \eqref{sys} is now written as
\begin{align}
\nonumber
y(t)= \phi^T(t) \theta + v(t),\  t=1,\dots, N
\end{align}
where $\phi(t)=[u(t-1),\cdots,u(t-n)]^T \in \mathbb{R}^n$,
and its matrix-vector form is
\begin{align}
&\hspace{2em}Y=\Phi \theta + V, ~\mbox{where}\label{firls}\\
\nonumber
&Y=[y(1)~y(n+2)\cdots~y(N)]^T\\
\nonumber
&\Phi=[\phi(1)~\phi(n+2)~\cdots~\phi(N)]^T\\
\nonumber
&V=[v(1)~v(n+2)~\cdots~v(N)]^T.
\end{align}
The well-known least squares (LS) estimator
\begin{subequations}
    \begin{align}
    \widehat{\theta}^{\rm LS}&=\argmin_{\theta \in \mathbb{R}^n}\|Y-\Phi\theta\|^2\label{lsa}\\
    &= (\Phi^T\Phi)^{-1}\Phi^T Y,\label{ls}
    \end{align}
\end{subequations} where $\|\cdot\|$ is the Euclidean norm,
is unbiased but may have large variance and mean square error (MSE) (e.g., when the input is low-pass filtered white noise). The large variance can be mitigated if  some bias is allowed and traded for smaller variance and smaller MSE.

One possible way to achieve this goal is to add a regularization
term $\sigma^2\theta^TP^{-1}\theta$ in the LS criterion \dref{lsa},
leading to the regularized least squares (RLS) estimator:
\begin{subequations}\label{eq:rls}
    \begin{align}
    \widehat{\theta}^{\rm R}=&\argmin_{\theta \in \mathbb{R}^n}\|Y-\Phi\theta\|^2 + \sigma^2\theta^TP^{-1}\theta\\
    =&P\Phi^T(\Phi P \Phi^T + \sigma^2 I_N)^{-1}Y \label{rls}
    \end{align}
\end{subequations}
where $P$ is positive semidefinite and is called the kernel matrix
($\sigma^2P^{-1}$ is often called the regularization matrix), and
$I_N$ is the $N$-dimensional identity matrix.

\begin{rem}\label{rmk1}
 As well known, the RLS estimator (\ref{rls}) has a Bayesian interpretation. Specifically,  assume that $\theta$ and $v(t)$ are independent and Gaussian distributed with
    \begin{align}\label{assumption}
    \theta \sim \mathscr{N}(0,P), \quad v(t) \sim \mathscr{N}(0,\sigma^2),
    \end{align} where $P$ is the prior covariance matrix. Then  $\theta$ and $Y$ are jointly Gaussian distributed and moreover, the posterior distribution of $\theta$ given $Y$ is
    \begin{align*}
    &\theta | Y  \sim \mathscr{N}(\widehat{\theta}^{\rm R},\widehat{P}^{\rm R})\\
    &\widehat{\theta}^{\rm R} = P\Phi^T(\Phi P \Phi^T + \sigma^2 I_N)^{-1}Y \\
    &\widehat{P}^{\rm R} = P - P\Phi^T(\Phi P \Phi^T + \sigma^2 I_N)^{-1}\Phi P.
    \end{align*}
 \end{rem}

Two types of MSE could be used to evaluate the performance of the
RLS estimator (\ref{rls}). The first one is the MSE related to the
impulse response reconstruction, see e.g.,
\cite{Chen2012,Pillonetto2015},
\begin{align}
{\rm MSEg}(P) = E(\|\widehat{\theta}^{\rm R}(P)-\theta_0\|^2),
\label{mseg}
\end{align}
where $E(\cdot)$ is the mathematical expectation and $\theta_0=
[g_1^0,\cdots,g_n^0]^T$ with $g_i^0$, $i=1,\dots,n$, defined in
(\ref{eq:truesys}). The second one is the MSE related to output
prediction, see e.g., \cite{Pillonetto2015},
\begin{align}
{\rm MSEy}(P)=E\left[\sum_{t=1}^N
\big(\phi^T(t)\theta_0+v^*(t)-\widehat{y}(t)\big)^2\right]\label{msey},
\end{align}
where $\widehat{y}(t)=\phi^T(t)\widehat{\theta}^{\rm R}(P)$ and $v^*(t)$ is an independent copy of the noise $v(t)$.
Interestinly, the two MSEs (\ref{mseg}) and (\ref{msey}) are related
with each other through \begin{align} {\rm MSEy}(P) \!=\! {\rm
Tr}\big(E(\widehat{\theta}^{\rm R} \!-\!\theta_0 ) (
\widehat{\theta}^{\rm R}\! -\!\theta_0)^T\Phi^T\Phi\big) \!+
\!N\sigma^2, \!\label{rgy}
\end{align} where $\rm
Tr(\cdot)$ is the trace of a square matrix. Moreover, they have
explicit expressions,
which are given in the following proposition.
\begin{prop}
    \label{prop1}
    For a given kernel matrix $P$, the two MSEs (\ref{mseg}) and (\ref{msey}) take the following form
\begingroup
\allowdisplaybreaks
\begin{align}
\nonumber
&{\rm MSEg}(P)=\|P\Phi^TQ^{-1}\Phi \theta_0 - \theta_0\|^2\\
&\hspace{6em}+\sigma^2{\rm Tr}(P\Phi^T Q^{-1} Q^{-T}\Phi P^T)
\label{msege}\\
\nonumber
&{\rm MSEy}(P)=\|\Phi P\Phi^TQ^{-1}\Phi \theta_0 - \Phi\theta_0\|^2 +N\sigma^2 \\
&\hspace{6em}+\!\sigma^2{\rm Tr}(\Phi P\Phi^T Q^{-1} Q^{-T}\Phi P^T\Phi^T  )\label{mseye}\\
& Q =  \Phi P \Phi^T + \sigma^2 I_N.\label{eq:Q}
\end{align}
\endgroup
\end{prop}

\subsection{RLS estimator can outperform LS estimator}
It is interesting to investigate whether the RLS estimator
(\ref{rls})  with a suitable choice of the kernel matrix $P$ can
have smaller MSEs (\ref{mseg}) and (\ref{msey}) than the LS
estimator (\ref{ls}). The answer is affirmative for MSEg
(\ref{mseg}) and for the ridge regression case, where
$P^{-1}=(\beta/\sigma^2) I_n$ with $\beta>0$,
\cite{Hoerl1970,Theobald1974}. In what follows, we further show that
this property also holds for more general $P$ for MSEg (\ref{mseg})
and MSEy (\ref{msey}).
\begin{prop}
    \label{th1}
Consider the RLS estimator \dref{rls} and the LS estimator
(\ref{ls}). Suppose that $P^{-1}=\beta A/\sigma^2$, where $\beta>0$
and $A$ is positive  semidefinite. Then for a given $A$, there exits
$\beta>0$ such that \dref{rls} has a smaller MSEg (\ref{mseg}) and
MSEy (\ref{msey}) than (\ref{ls}). Moreover, if $A$ is positive
definite, then \dref{rls} has a smaller MSEg (\ref{mseg}) and MSEy
(\ref{msey}) than (\ref{ls}) whenever
$0<\beta<2\sigma^2/(\theta_0^T\!A\theta_0)$.
\end{prop}
Proposition \ref{th1} shows that for any data length $N$, the RLS
estimator (\ref{rls}) can have smaller MSEg (\ref{mseg}) and MSEy
(\ref{msey}) than the LS estimator \dref{ls} with a sufficiently
small regularization ``in any direction"  and this merit motivates
to further explore the potential of the RLS estimator (\ref{rls}) by
careful design of the kernel matrix $P$.

\section{Design of Kernel Matrix and Hyperparameter Estimation}
\label{sec3}

The regularization method has two core issues: kernel matrix design,
namely parameterization of the kernel matrix by a parameter vector,
called hyperparameter, and the hyperparameter estimation.

\subsection{Parametrization of Kernel Matrix}
For efficient regularization, the kernel matrix $P$ has to be chosen
carefully. It is typically done by postulating a parameterized family of matrices
\begin{align}
\label{eq:12}
P(\eta), \quad \eta \in \Omega\subset \mathbb R^p,
\end{align}
where $\eta$ is called the \emph{hyperparameter} and the feasible
set $\Omega$ of $\eta$ is assumed to be compact. The choice of
parameterization is a trade-off of the same kind as the choice of
model class in identification: On one hand it should be a large and
flexible class to allow as much benefits from regularization as
possible. On the other hand, a large set requires larger dimensions
of $\eta$, and the estimation of these comes with their own
penalties (much in the spirit of the Akaike's criterion). Since $P$
is the prior covariance of the true impulse response, the prior
knowledge of the underlying system to be identified, e.g.,
exponential stability and smoothness, should be embedded in the
parameterized matrix $P(\eta)$.

A popular way to achieve this goal is through a parameterized
positive semidefinite kernel function. So far, several kernels have
been invented, such as the stable spline (SS) kernel
\citep{Pillonetto2010}, the diagonal correlated (DC) kernel and the
tuned-correlated (TC) kernel \citep{Chen2012}, which are defined as
follows:
\begingroup
\allowdisplaybreaks
\begin{align}
\nonumber
&{\rm SS}:~~P_{kj}(\eta) = c\left(\frac{\alpha^{k+j+\max(k,j)}}2 -\frac{\alpha^{3\max(k,j)}}6 \right),\\
&~~~~~~~\eta = [c,\alpha]\in\Omega=\{c\geq 0,0\leq \alpha \leq 1\};\\
\nonumber
&{\rm DC}:~~P_{kj}(\eta) = c\alpha^{(k+j)/2}\rho^{|j-k|},\\
&~~\eta = [c,\alpha,\rho]\in \Omega=\{c\geq 0,0\leq \alpha \leq 1, |\rho|\leq 1\};\\
\nonumber
&{\rm TC}:~~P_{kj}(\eta) = c\alpha^{\max(k,j)},\\
&~~~~~~~\eta = [c,\alpha]\in\Omega=\{c\geq 0,0\leq \alpha \leq 1\}.\label{tc}
\end{align}
\endgroup

\subsection{Hyperparameter Estimation}
\label{sec:hyper_estimator}
Once a parameterized family of the kernel matrix $P(\eta)$ has been
chosen, the task is to estimate, or ``tune'', the hyperparameter
$\eta$  based on the data.

Several methods are suggested in the literature, see e.g., Section
14 of \cite{Pillonetto2014}, including the empirical Bayes (EB) and
SURE methods.
The EB method uses the Bayesian interpretation in Remark \ref{rmk1}.
Under the assumption (\ref{assumption}), it follows that $Y$ is
Gaussian with mean zero and covariance matrix $\Phi^TP(\eta)\Phi +
\sigma^2 I_N$. As a result, it is possible to estimate the
hyperparameter $\eta$ by maximizing the (marginal) likelihood of
$Y$, i.e.,
\begingroup
\allowdisplaybreaks
\begin{align}
{\rm EB:}\ & {\widehat\eta}_{\rm EB} = \argmin_{\eta \in \Omega} \mathscr{F}_{\rm EB}(P(\eta)),\label{thml}\\
& \mathscr{F}_{\rm EB}(P) = Y^TQ^{-1} Y + \log\det (Q).\label{mlcf}
\end{align} where $Q$ is defined in \eqref{eq:Q} and $\det(\cdot)$
denotes the determinant of a square matrix.
\endgroup
The SURE method first constructs a Stein's unbiased risk estimator
(SURE) of the MSE and then estimates the hyperparameter by
minimizing the constructed estimator. Two variants of the SURE
method were considered in \cite{Pillonetto2015}, which construct
the SUREs for ${\rm MSEg}(P)$ in (\ref{msege}) and ${\rm MSEy}(P)$
in (\ref{mseye}), and are referred to as SUREg and SUREy,
respectively:
\begin{align}
\nonumber
 \mathscr{F}_{\rm Sg}(P) &= \|\widehat{\theta}^{\rm LS} - \widehat{\theta}^{\rm R}(P) \|^2+\sigma^2{\rm Tr}\big(2R^{-1}\!-(\Phi^T\Phi)^{-1}\big)\\
\nonumber &=\sigma^4Y^TQ^{-T}\Phi (\Phi^T\Phi)^{-2}\Phi^T Q^{-1} Y\\
&\qquad+\sigma^2{\rm Tr}\big(2R^{-1}\!-(\Phi^T\Phi)^{-1}\big)\label{sgcf}\\
\mathscr{F}_{\rm Sy}(P)\! &=\! \|Y\!\!-\Phi\widehat{\theta}^{\rm
R}(P)\|^2
\!+\!2\sigma^2{\rm Tr}\big(\Phi P\Phi^TQ^{-1}\big) \nonumber\\
&=\!\sigma^4Y^TQ^{-T} Q^{-1} Y\!+\!2\sigma^2{\rm Tr}\big(\Phi
P\Phi^TQ^{-1}\big)\!\!\!\!\label{sycf}\\ R &= \Phi^T  \Phi +
\sigma^2 P^{-1}.\label{eq:R}
\end{align}
Then the hyperparameter $\eta$ is estimated by minimizing the SUREg (\ref{sgcf}) and SUREy (\ref{sycf}): \begin{align}
{\rm SUREg:}\ &{\widehat\eta}_{\rm Sg} = \argmin_{\eta \in \Omega} \mathscr{F}_{\rm Sg}(P(\eta)),\label{thsg}\\
{\rm SUREy:}\ &{\widehat\eta}_{\rm Sy} = \argmin_{\eta \in \Omega}
\mathscr{F}_{\rm Sy}(P(\eta)). \label{thsy}
\end{align}
In the following sections, we will study the properties of the above
three estimators EB, SUREg and SUREy. To set reference for these
estimators, we introduce their corresponding Oracle counterparts
that depend on the true impulse response $\theta_0$:
\begingroup
\allowdisplaybreaks
\begin{align}
&{\rm MSEg:}\ {\widehat\eta}_{\rm MSEg} = \argmin_{\eta \in \Omega} E [\mathscr{F}_{\rm Sg}(P(\eta)]\nonumber \\&\qquad\qquad\qquad\ =\argmin_{\eta \in \Omega} {\rm MSEg}(P(\eta)),\label{okmseg}\\
&{\rm MSEy:}\ {\widehat\eta}_{\rm MSEy}  = \argmin_{\eta \in \Omega} E[\mathscr{F}_{\rm Sy}(P(\eta))]\nonumber\\&\qquad\qquad\qquad\ = \argmin_{\eta \in \Omega} {\rm MSEy}(P(\eta)),\label{okmsey}\\
&{\rm EEB:}\ {\widehat\eta}_{\rm EEB} = \argmin_{\eta \in \Omega} E [\mathscr{F}_{\rm EB}(P(\eta))]\nonumber\\
&\qquad\qquad\qquad\ = \argmin_{\eta \in \Omega} {\rm EEB}(P(\eta)), \label{okmseml}\\
&{\rm EEB}(P)=\theta_0^T \Phi^T Q^{-1} \Phi \theta_0 \!+
\!\sigma^2{\rm Tr}( Q^{-1}  ) + \log\det (Q),\label{eebe}
\end{align}
\endgroup
where ${\rm MSEg}(P)$ and ${\rm MSEy}(P)$ are defined in
(\ref{msege}) and (\ref{mseye}), respectively.

The hyperparameter estimators (\ref{okmseg}) and (\ref{okmsey}) give
the optimal hyperparameter estimates in the corresponding MSE sense
and thus provide reference when evaluating the performance of
hyperparameter estimators.

\begin{rem}
Among these hyperparameter estimators, only SUREg \eqref{thsg}
depends on $(\Phi^T\Phi)^{-1}$. When $(\Phi^T\Phi)^{-1}$ is
ill-conditioned, SUREg \eqref{thsg} should be avoided for
hyperparameter estimation. One may also note that
$(\Phi^T\Phi)^{-1}$ in the second term is independent of $P$ and
thus can actually be removed in the calculation.
\end{rem}

\begin{rem} It is interesting to note that
the first terms of $\mathscr{F}_{\rm Sg}(P),\mathscr{F}_{\rm
Sy}(P),$ and $\mathscr{F}_{\rm EB}(P)$ given in \dref{sgcf},
\dref{sycf}, and \dref{mlcf} contain the same factors $Y$ and
$Q^{-1}$. Moreover, similar to \dref{rgy}, $\mathscr{F}_{\rm Sg}(P)$
and $\mathscr{F}_{\rm Sy}(P)$ are related with each other through
\begin{align}
\nonumber \mathscr{F}_{\rm Sy}(P) &= {\rm Tr} \big\{\big[
(\widehat{\theta}^{\rm LS} \!- \!\widehat{\theta}^{\rm R}(P))
(\widehat{\theta}^{\rm LS} \!- \! \widehat{\theta}^{\rm R}(P))^T\\
\nonumber &\hspace{2em}+
\sigma^2(2R^{-1}\!- \!(\Phi^T\Phi)^{-1} ) \big]\Phi^T\Phi\big\}\\
& ~~~~+\!\underbrace{Y^T\Phi (\Phi^T\Phi)^{-1}\Phi^T  Y \!-\! Y^TY
-n\sigma^2.}_{\mbox{independent of the kernel matrix }
P}\label{rsgy}
\end{align}
\end{rem}
In what follows, we will investigate the properties of the
hyperparameter estimators EB, SUREg, and SUREy and their
corresponding Oracle estimators EEB, MSEg and MSEy. Before
proceeding to the details, we make, without loss of generality, the
following assumption.

\begin{assum}\label{ass:1}
The optimal hyperparameter estimates ${\widehat\eta}_{\rm Sg}$,
${\widehat\eta}_{\rm Sy}$, ${\widehat\eta}_{\rm EB}$,
${\widehat\eta}_{\rm MSEg}$, ${\widehat\eta}_{\rm MSEy}$ and
${\widehat\eta}_{\rm EEB}$ are interior points of $\Omega$.
\end{assum}

\begin{rem}
To justify Assumption \ref{ass:1}, we take the DC kernel as an
example. For the case where either $c=0$ or $\alpha = 0$,
$P(\eta)=0$ and thus (\ref{rls}) is trivially $0$. For the case
where $\alpha = 1$, this violates the stability of the system. For
the case where $|\rho| = 1$, the coefficients of the impulse
response is perfectly positive or negative correlated, but this is
impossible for a stable system. In fact, more formal justification
regarding this assumption can be found on \cite[p.
115]{Pillonetto2015}, which shows that the measure of the set
containing all optimal estimates lying on the boundary of $\Omega$
is zero and thus can be neglected when making almost sure
convergence statement.
\end{rem}

\section{Properties of Hyperparameter Estimators: Finite Data Case}
\label{sec4}

In this section, focusing on the finite data case we first give the
first order optimality conditions of the hyperparameter estimators
and then we consider two special cases for which closed-form
expressions of the optimal hyperparameter estimates are available.

\subsection{First Order Optimality Conditions}
The optimal hyperparameter estimates $\widehat{\eta}_{\rm Sg}$,
$\widehat{\eta}_{\rm Sy}$, and $\widehat{\eta}_{\rm EB}$ in
\dref{thsg}, \dref{thsy}, and \dref{thml} should satisfy the first
order optimality conditions if they are interior points of $\Omega$.
For convenience, we let $\mathscr{C}$ to denote one of the following
estimation criteria $\mathscr{F}_{\rm Sg}$, $\mathscr{F}_{\rm Sy}$,
$\mathscr{F}_{\rm EB}$, MSEg, MSEy or EEB. Then the corresponding
optimal hyperparameter estimate is a root of the system of
equations:
\begin{align}
\frac{\partial \mathscr{C}(P(\eta))}{\partial \eta} =0.\label{scsg}
\end{align}
By the chain rule of compound functions, we have
\begingroup
\allowdisplaybreaks
\begin{align}
{\rm Tr}\left(\frac{\partial \mathscr{C}(P)}{\partial P}
\Big(\frac{\partial P(\eta)}{\partial \eta_i}\Big)^T\right)
=0,~1\leq i\leq p.\label{sksg}
\end{align}
where
the symmetry of $P$ is not considered, that is, the elements of $P$ are treated independently.
\endgroup
Clearly, the term $\frac{\partial \mathscr{C}(P)}{\partial P}$ is
irrespective of the parameterization of $P$ and its explicit
expressions for the estimation criteria \dref{sgcf}, \dref{sycf},
and \dref{mlcf} are available.
\begin{prop}
    \label{thm3}
    The first order partial derivatives of \dref{sgcf}, \dref{sycf}, and \dref{mlcf} with respect to $P$
    are, respectively,
    \begingroup
    \allowdisplaybreaks
    \begin{align}
    \nonumber
   &\frac{\partial \mathscr{F}_{\rm Sg}(P)}{\partial P}
   \!=\!-2\sigma^4
   \Phi^TQ^{-T} \Phi(\Phi^T\Phi)^{-2}\Phi^TQ^{-1}YY^TQ^{-T}\Phi\\
   &\hspace{7em}+ 2\sigma^4H^{-T}\overline{H}^{-T} \label{dsg}\\
      \nonumber
    &\frac{\partial \mathscr{F}_{\rm Sy}(P)}{\partial P}
    \!=\!-2\sigma^4 \Phi^T Q^{-T} Q^{-1}YY^T Q^{-T} \Phi\\
    &\hspace{7em}+2\sigma^4\Phi^T Q^{-T} Q^{-T}\Phi \label{dsy}\\
        &\frac{\partial \mathscr{F}_{\rm EB}(P)}{\partial P}
        =-\Phi^TQ^{-T}YY^TQ^{-T}\Phi
        +\Phi^TQ^{-T} \Phi, \label{dml}\\
& H = P\Phi^T\Phi + \sigma^2I_n,\ \overline{H}=\Phi^T\Phi P +
\sigma^2I_n.\label{eq:H}
    \end{align}
    \endgroup
\end{prop}
Similarly, the partial derivatives of ${\rm MSEg}(P)$, ${\rm
MSEy}(P)$, and ${\rm EEB}(P)$ with respect to $P$ are also
available.
\begin{prop}
    \label{thmb1}
    The first order partial derivatives of \eqref{msege}, \eqref{mseye}, and \eqref{eebe} with respect to $P$ are,
    respectively,
    \begingroup
    \allowdisplaybreaks
    \begin{align}
    \nonumber
    \frac{{\partial {\rm MSEg}(P) }}{\partial P}
    &=-2\sigma^4H^{-T}H^{-1}
    \theta_0 \theta_0^T
    \Phi^T
    Q^{-T}\Phi\\
    &~~~~+2\sigma^4
    H^{-T}H^{-1}P\Phi^T
    Q^{-T}\Phi \label{dmsege}\\
        \nonumber
    \frac{{\partial {\rm MSEy}(P) }}{\partial P}
    &=-2\sigma^4\Phi^T Q^{-T} Q^{-1}\Phi
    \theta_0 \theta_0^T
    \Phi^T Q^{-T}\Phi \\
    &~~~~+2\sigma^4
    \Phi^T Q^{-T} Q^{-1}
    \Phi P \Phi^T
    Q^{-T}\Phi\label{dmseye}\\
    \nonumber
    \frac{\partial {\rm EEB}(P)}{\partial P}
    &=-\Phi^T Q^{-T} \Phi\theta_0\theta_0^T\Phi^T Q^{-T} \Phi\\
    &~~~~+\Phi^T Q^{-T} \Phi P^T\Phi^T Q^{-T} \Phi.\label{dmsemle}
    \end{align} where $H$ is defined in \eqref{eq:H}.
    \endgroup
\end{prop}

In order to better expose the relation among the partial derivatives
derived in Propositions \ref{thm3} and \ref{thmb1}, we define
\begin{align}\label{eq:S}
S = P + \sigma^2(\Phi^T  \Phi) ^{-1}.
\end{align}
With the use of (\ref{eq:S}) and the identities \dref{id6}--\dref{id5} in the appendix, we rewrite the partial derivatives
derived in Propositions \ref{thm3} and \ref{thmb1} as follows.
\begin{cor}
    \label{cor4}
    The partial derivatives derived in Propositions \ref{thm3} and \ref{thmb1}
    can be rewritten as follows:
    \begingroup
    \allowdisplaybreaks
    \begin{align}
    &\frac{{\partial {\rm MSEg}(P) }}{\partial P}
    =2\sigma^4
    S^{-T}
    (\Phi^T \Phi )^{-2}
    S^{-1}
    (P -  \theta_0\theta_0^T )
    S^{-T} \label{dmsegec}\\
    &\frac{\partial \mathscr{F}_{\rm Sg}(P)}{\partial P}
    \!=\!2\sigma^4S^{-T}(\Phi^T\Phi)^{-2}S^{-1}
    \big(S - \widehat{\theta}^{\rm LS}(\widehat{\theta}^{\rm LS})^T\big)
    S^{-T}\label{dsgc}\\
    &\frac{{\partial {\rm MSEy}(P) }}{\partial P}
    =2\sigma^4
    S^{-T}
    (\Phi^T \Phi )^{-1}
    S^{-1}
    (P -  \theta_0\theta_0^T)
    S^{-T}\!\!\!\! \label{dmseyec}  \\
    &\frac{\partial \mathscr{F}_{\rm Sy}(P)}{\partial P}
    \!=\!2\sigma^4S^{-T}(\Phi^T\Phi)^{-1}S^{-1}
    \big(S - \widehat{\theta}^{\rm LS}(\widehat{\theta}^{\rm LS})^T\big)
    S^{-T}\label{dsyc}  \\
    &\frac{\partial {\rm EEB}(P)}{\partial P}
    =S^{-T}
    (P^T- \theta_0\theta_0^T)
    S^{-T}\label{dmsemlec}
\\
    &\frac{\partial \mathscr{F}_{\rm EB}(P)}{\partial P}
    =S^{-T}\big(S^T - \widehat{\theta}^{\rm LS}(\widehat{\theta}^{\rm LS})^T\big)
    S^{-T}. \label{dmlc}
    \end{align}
    \endgroup
\end{cor}
It follows from Corollary \ref{cor4} that the difference between the
partial derivatives of $\mathscr{F}_{\rm Sg}(P),\mathscr{F}_{\rm
Sy}(P),\mathscr{F}_{\rm EB}(P)$ and that of their Oracle
counterparts is that the factor $S - \widehat{\theta}^{\rm
LS}(\widehat{\theta}^{\rm LS})^T$ is replaced by $P-
\theta_0\theta_0^T$.
Moreover, the difference between the partial derivative of
$\mathscr{F}_{\rm Sg}(P)$ and that of $ \mathscr{F}_{\rm Sy}(P)$ is
that there is one extra factor $(\Phi^T\Phi)^{-1}$. The difference
between the first order derivative of $\mathscr{F}_{\rm Sy}(P)$ and
that of $ \mathscr{F}_{\rm EB}(P)$ is that there is one extra factor
$2\sigma^4(\Phi^T\Phi)^{-1}S^{-1}=2\sigma^4H^{-1}$. The above
relations extend to the partial derivatives of their Oracle
counterparts.
\begin{rem} It is important to note from Propositions \ref{thm3} and \ref{thmb1}
that only the first term of $\frac{\partial \mathscr{F}_{\rm
Sg}(P)}{\partial P}$ depends on the possibly ill-conditioned
$(\Phi^T\Phi)^{-1}$. With the use of $S$ in \eqref{eq:S}, all
partial derivatives of the hyperparameter estimators seemingly
depend on the possibly ill-conditioned term $(\Phi^T\Phi)^{-1}$.
However, it should be stressed that the partial derivatives derived
in Corollary \ref{cor4} are not intended for numerical calculation
but for theoretical analysis and for better exposition of the
relation among the partial derivatives derived in Propositions
\ref{thm3} and \ref{thmb1}.
\end{rem}

\begin{rem}
The kernel matrix $P$ is in general assumed to be symmetric. In this
case, we have $S^T = S$ and thus the partial derivatives derived in
Corollary \ref{cor4} can be simplified accordingly.
\end{rem}

Setting $\frac{{\partial {\rm MSEg}(P) }}{\partial P}=0$,
$\frac{{\partial {\rm MSEg}(P) }}{\partial P}=0$,  and
$\frac{\partial {\rm EEB}(P)}{\partial P}=0$ in Corollary \ref{cor4}
leads to the next proposition.
\begin{prop} \label{prop5}
The optimal kernel matrix that minimizes ${\rm MSEg}(P)$,  ${\rm
MSEy}(P)$, and ${\rm EEB}(P)$ without structure constraints on $P$
is \begin{align}
    P=\theta_0\theta_0^T. \label{ok}
    \end{align}
\end{prop}
It was found in \cite{Chen2012} that \eqref{ok} minimizes the MSE
matrix $E(\widehat{\theta}^{\rm R}-\theta_0)(\widehat{\theta}^{\rm
R}-\theta_0)^T$ in the matrix sense. Here we further find that
\eqref{ok} is optimal for ${\rm MSEg}(P)$, ${\rm MSEy}(P)$ and ${\rm
EEB}(P)$, and for any data length $N$.

\begin{rem}
It seems that $S - \widehat{\theta}^{\rm LS}(\widehat{\theta}^{\rm
LS})^T=0$, i.e., $P=\widehat{\theta}^{\rm LS}(\widehat{\theta}^{\rm
LS})^T-\sigma^2(\Phi^T  \Phi)^{-1}$ is a possible candidate for the
optimal matrix minimizing SUREg($P$), SUREy($P$), and EB($P$).
However, this is not true, since this kernel matrix would make
$S=\widehat{\theta}^{\rm LS}(\widehat{\theta}^{\rm LS})^T$ singular
and SUREg($P$), SUREy($P$), and EB($P$) take the value of $-\infty$.
\end{rem}
In general, there is no explicit expression of these hyperparameter estimators.
However, there exist some
specific cases, for which it is possible to derive the explicit
solution based on Corollary \ref{cor4}. In the following, we
consider two special cases.

\subsection{Ridge Regression with $\Phi^T\Phi = NI_n$}
We let $P(\eta)=\eta I_n$ with $\eta\geq 0$ and assume $\Phi^T\Phi =
NI_n$. Then we have the following result.
\begin{prop}
    \label{prop6}
    Consider $P(\eta)=\eta I_n$ with $\eta\geq 0$.
    Further assume that $\Phi^T\Phi = NI_n$.
    Then we have \begin{align}
\widehat{\eta}_{\rm Sg}=\widehat{\eta}_{\rm Sy}=\widehat{\eta}_{\rm
EB}= \max\Big(0,\frac{(\widehat{\theta}^{\rm LS})^T
        \widehat{\theta}^{\rm LS} }{n}
    - \frac{\sigma^2}{N }\Big).\label{rro}
    \end{align}
    Moreover,  \begin{align}
{\widehat\eta}_{\rm MSEg}={\widehat\eta}_{\rm
MSEy}={\widehat\eta}_{\rm EEB}= \theta_0^T\theta_0/n.
        \end{align}
\end{prop}
\begin{rem}
It is worth noting that the optimal hyperparameter
$\theta_0^T\theta_0/n$ holds for any $N$. Moreover, one has
\begin{align*}
&{\rm
MSEg}(\theta_0^T\theta_0/nI_n)=\frac{n\sigma^2}{N+n\sigma^2/(\theta_0^T\theta_0)}
<\frac{n\sigma^2}{N},
\end{align*}
where $n\sigma^2/N$ is equal to the MSEg of the LS estimator
(\ref{ls}). This means that  the ridge regression with
$P=\theta_0^T\theta_0/nI_n$ has a smaller MSEg than the LS estimator
(\ref{ls}) when $\Phi^T\Phi=NI_n$. Finally, \dref{rro} is a
consistent estimator of $\theta_0^T\theta_0/n$ if
$\widehat{\theta}^{\rm LS}\xra{}\theta_0$ as $N\xra{}\infty$.

\end{rem}

\subsection{Diagonal Kernel Matrix with $\Phi^T\Phi = NI_n$}

We let $P(\eta)$ be a diagonal kernel matrix (in this case we have
$p=n$.), i.e.,
\begin{align}
P(\eta)=\diag[\eta_1,\cdots,\eta_n]~\mbox{with}~\eta_i\geq 0,~1\leq
i\leq n\label{dk}.
\end{align} where $\eta_1,\cdots,\eta_n$ are the main diagonal elements of the diagonal matrix $\diag[\eta_1,\cdots,\eta_n]$.
Then under the assumption $\Phi^T\Phi = NI_n$, we have the following
result.

\begin{prop}
    \label{prop7}
Consider  $P(\eta)$ in \eqref{dk}. Further assume that $\Phi^T\Phi =
NI_n$. Then we have
    \begin{align}
\widehat{\eta}_{\rm Sg}=\widehat{\eta}_{\rm Sy}=\widehat{\eta}_{\rm
EB}\nonumber&= \left[\max\{0,\widehat{g}_1^2 \!-\!
\sigma^2/N\},\right.\\&\left.\cdots,\max\{0,\widehat{g}_n^2 \!-\!
\sigma^2/N\}\right]^T
\label{okml}
    \end{align}
    where $\widehat{g}_i$ is the $i$-th element of the LS estimate \eqref{ls}, $i=1,\dots,n$.
Moreover,
    \begin{align}
{\widehat\eta}_{\rm MSEg}
\!=\!{\widehat\eta}_{\rm
MSEy}
\!=\!{\widehat\eta}_{\rm EEB}
\!=\!
\left[(g_1^{0})^2,\cdots,(g_n^{0})^2\right]^T\!. \label{gok}
    \end{align}
\end{prop}
\begin{rem}
In the papers \citep{Aravkin2012,Aravkin2014}, the linear model
\dref{firls} but with a slightly different setting is considered,
where the parameter $\theta$ is partitioned into $m$ sub-vectors
$\theta=[\theta^{(1)^T},\cdots,\theta^{(m)^T}]^T$ and the dimension
of $\theta^{(i)}$ is $n_i$ so that $n=\sum_{i=1}^{m} n_i$. In
addition, the prior distribution of $\theta^{(i)}$ is set to be
$\mathscr{N}(0,\eta_iI_{n_i})$ and $\eta_i$ is an independent and
identically distributed exponential random variable with probability
density $p_\gamma(\eta_i) = \gamma\exp(-\gamma \eta_i)\chi(\eta_i)$
where $\gamma$ is a positive scalar and $\chi(t)=1$ for $t\geq 0$
and 0 otherwise. Under the setting given above, the solution
maximizing the marginal posterior density of $\eta$ given the data
and the optimal solution of the MSEg are derived in
\cite{Aravkin2012,Aravkin2014} when $\Phi^T\Phi = NI_n$. When
$n_i=1$ for $i=1,\cdots,m$ and $\gamma=0$, their solutions become
\dref{okml} and \dref{gok}, respectively. In contrast, we study here
the SUREg, SUREy, MSEy, and EEB estimators other than the EB and
MSEg estimators and find their solutions are the same under the
simplified setting, respectively. Clearly, $\max\{0,\widehat{g}_i^2
- \sigma^2/N\}$ is a  consistent estimator of $(g_i^{0})^2$,
$i=1,\dots,n$.
\end{rem}

\section{Properties of Hyperparameter Estimators: Infinite Data Case}
\label{sec5} In this section, we investigate the asymptotic
properties of these hyperparameter estimators. For this purpose, it
is useful to first consider the asymptotic property of the partial
derivatives derived in Corollary \ref{cor4}. Noting the finding of
Corollary \ref{cor4} under \eqref{dmlc} and that $S -
\widehat{\theta}^{\rm LS}(\widehat{\theta}^{\rm LS})^T$ converges to
$P - \theta_0\theta_0^T$ under proper conditions, we can derive the
following Proposition.
\begin{prop}
    \label{prop8} Consider the partial derivatives derived in
    Corollary \ref{cor4}. Assume that $P$ is nonsingular and $\Phi^T\Phi/N \xra{} \Sigma$ almost surely
 as $N\xra{}\infty$, where $\Sigma$ is positive definite.
    Then we have as $N\xra{}\infty$
    \begingroup
    \allowdisplaybreaks
    \begin{align}
    &N^2\frac{{\partial {\rm MSEg}(P) }}{\partial P}
    \!\xra{}\!2\sigma^4
    P^{-T}
    \Sigma^{-2}
    P^{-1}
    (P \!- \! \theta_0\theta_0^T )
    P^{-T}\! \!\label{dmsegel}\\
    &N^2\frac{\partial \mathscr{F}_{\rm Sg}(P)}{\partial P}
    \!\xra{}\!2\sigma^4
    P^{-T}
    \Sigma^{-2}
    P^{-1}
    (P\!- \!  \theta_0\theta_0^T )
    P^{-T} \label{dsgel}\\
    &N\frac{{\partial {\rm MSEy}(P) }}{\partial P}
    \!\xra{}\!2\sigma^4
    P^{-T}
    \Sigma^{-1}
    P^{-1}
    (P \!- \!  \theta_0\theta_0^T)
    P^{-T} \label{dmseyel}\\
    &N\frac{\partial \mathscr{F}_{\rm Sy}(P)}{\partial P}
    \!\xra{}\!2\sigma^4
    P^{-T}
    \Sigma^{-1}
    P^{-1}
    (P \!- \! \theta_0\theta_0^T)
    P^{-T} \label{dsyel}\\
    &\frac{\partial {\rm EEB}(P)}{\partial P}
    \!\xra{}\!P^{-T}
    (P^T- \theta_0\theta_0^T)
    P^{-T}\\
    &\frac{\partial \mathscr{F}_{\rm EB}(P)}{\partial P}
    \!\xra{}\!P^{-T}
    (P^T\!- \! \theta_0\theta_0^T)
    P^{-T}
    \end{align}
    \endgroup
    almost surely.
\end{prop}
Proposition \ref{prop8} shows that the three pairs,
$N^2\frac{{\partial {\rm MSEg}(P) }}{\partial P}$ and
$N^2\frac{\partial \mathscr{F}_{\rm Sg}(P)}{\partial P}$, and
$N\frac{{\partial {\rm MSEy}(P) }}{\partial P}$ and
$N\frac{\partial \mathscr{F}_{\rm Sy}(P)}{\partial P}$, and
$\frac{{\partial {\rm EEB}(P) }}{\partial P}$ and $\frac{\partial
\mathscr{F}_{\rm EB}(P)}{\partial P}$, have respectively the same
limit as $N$ goes to $\infty$. This observation motivates to explore
if this property also holds for the estimation criteria of these
hyperparameter estimators. The answer is affirmative and we have the
following result.
\begin{prop}
    \label{thm11} Consider the hyperparameter estimation criteria SUREg \dref{sgcf},  SUREy \dref{sycf}, and
    EB \dref{mlcf}, and their corresponding Oracle counterparts
    MSEg \eqref{msege}, MSEy \eqref{mseye}, and EEB \eqref{eebe}.
    Assume that $P$ is nonsingular and $\Phi^T\Phi/N \xra{} \Sigma$ almost surely
 as $N\xra{}\infty$, where $\Sigma$ is positive definite.
    Then we have as $N\xra{}\infty$
    \begingroup
    \allowdisplaybreaks
    \begin{align}
    &N^2({\rm MSEg}(P) - \sigma^2{\rm Tr}( (\Phi^T \Phi)^{-1}))
    \xra{}W_g(P,\Sigma,\theta_0)\!\label{cl1}\\
    &N^2(\mathscr{F}_{\rm Sg}(P) - \sigma^2{\rm Tr}( (\Phi^T \Phi)^{-1}))
    \xra{}W_g(P,\Sigma,\theta_0), \label{cl2}\\
    &N({\rm MSEy}(P) - (n+N)\sigma^2)\xra{}W_y(P,\Sigma,\theta_0) \label{cl3}\\
    \nonumber
    &N(\mathscr{F}_{\rm Sy}(P)
    + Y^T\Phi(\Phi^T\Phi)^{-1}\Phi^TY
    -Y^TY
    - 2n\sigma^2)\\
    &\hspace{4.1cm}\xra{}W_y(P,\Sigma,\theta_0), \label{cl4}\\
    \nonumber
    &{\rm EEB}(P)- (N-n) \\
    &\hspace{0.4cm} - (N\!-\!n)\log \sigma^2 \! - \! \log \det(\Phi^T\Phi)\!\xra{}\!W_{\rm B}(P,\theta_0),\label{cl5}\\
    \nonumber
    &\mathscr{F}_{\rm EB}(P)
    + Y^T\Phi(\Phi^T\Phi)^{-1}\Phi^TY/\sigma^2
    -Y^TY/\sigma^2\\
    &\hspace{0.4cm} - (N\!-\!n)\log \sigma^2 \! - \! \log \det(\Phi^T\Phi)\!\xra{}\!W_{\rm B}(P,\theta_0),\label{cl6}
    \end{align}
    \endgroup
    almost surely, where
    \begingroup
    \allowdisplaybreaks
 \begin{align}
W_g(P,\Sigma,\theta_0)&=
\sigma^4\theta_0^TP^{-T}\Sigma^{-2}P^{-1}\theta_0\nonumber\\
&~~~~
-2\sigma^4{\rm Tr}
\big( \Sigma^{-1}P^{-1}\Sigma^{-1}\big),  \label{msedc}\\
\nonumber
W_y(P,\Sigma,\theta_0)& =
\sigma^4\theta_0^TP^{-T}\Sigma^{-1}P^{-1}\theta_0\\
&~~~~-2\sigma^4{\rm Tr}
\big(
\Sigma^{-1}P^{-1}\big),\label{msedy}\\
W_{\rm B}(P,\theta_0)& =
\theta_0^TP^{-1}\theta_0 +\log \det(P). \label{ebdc}
\end{align}
    \endgroup
\end{prop}
\begin{rem}
For these hyperparameter estimation criteria,
$W_g(P,\Sigma,\theta_0)$, $W_y(P,\Sigma,\theta_0)$ and
$W_B(P,\theta_0)$ contain all information about the
asymptotic benefits of regularization: how it depends on any kernel
matrix $P$, any true impulse response vector $\theta_0$ and any
stationary properties of the input covariance matrix $\Sigma$.
\end{rem}
Proposition \ref{thm11} enable us to derive asymptotic properties of
these hyperparameters estimator for any parameterization $P(\eta)$
of the kernel matrix. Moreover, it also implies that the estimators
$\widehat{\eta}_{\rm Sg}$, $\widehat{\eta}_{\rm Sy}$, and
$\widehat{\eta}_{\rm EB}$ possibly share the same limits with their
corresponding Oracle counterparts $\widehat{\eta}_{\rm MSEg}$,
$\widehat{\eta}_{\rm MSEy}$, and $\widehat{\eta}_{{\rm EEB}}$,
respectively.

To state the result, we need an extra assumption. It is worth to
note that the limit functions $W_g(P(\eta),\Sigma,\theta_0)$,
$W_y(P(\eta),\Sigma,\theta_0)$ and $W_{\rm
    B}(P(\eta),\theta_0)$ may not have a unique global minimum, respectively.
In this case,
    the analysis of how minimizing elements of a sequence of
    functions $M_N(\eta)$
    converge to the minimizing element of the limit function $\lim
    M_N(\eta)$, i.e.,
    \begin{align}\label{eq:Meta}
    ``\lim  \arg\min M_N(\eta) = \arg\min \lim M_N(\eta)",
    \end{align} where $M_N(\eta)$ denotes any function on the left
    handside of ``$\rightarrow$'' in \eqref{cl1} to \eqref{cl6},
    follows the same idea as for prediction error identification
    methods, see, e.g. Lemma 8.2 and Theorem 8.2 in
    \cite{Ljung1999}. Accordingly, it is useful in this context to let ``$\arg\min$'' denote the set of
    minimizing arguments in case where $W_g(P(\eta),\Sigma,\theta_0)$, $W_y(P(\eta),\Sigma,\theta_0)$ and
$W_{\rm B}(P(\eta),\theta_0)$ do not have a unique global minimum,
respectively,:
    \begin{align}\label{eq:D}
    \arg\min_{\eta \in \Omega} M(\eta) = \big\{\eta|\eta \in
    \Omega, M(\eta)=\min_{\eta'\in \Omega }M(\eta')\big\},
    \end{align}
where $M(\eta)$ could be any one of $W_g(P(\eta),\Sigma,\theta_0)$,
$W_y(P(\eta),\Sigma,\theta_0)$ and $W_{\rm B}(P(\eta),\theta_0)$.

Now we define
\begin{align}
\label{osg} &\eta_g^*=\arg\min_{\eta \in
\Omega}W_g(P(\eta),\Sigma,\theta_0),\\
&\eta_{\rm y}^*
=\argmin_{\eta \in \Omega} W_y(P(\eta),\Sigma,\theta_0),\label{osy}\\
&\eta_{\rm B}^* =\argmin_{\eta \in \Omega} W_{\rm
B}(P(\eta),\theta_0).\label{oml}
\end{align}
and the assumption we need can be stated as follows.
\begin{assum}\label{ass:2}
    The sets $\eta_g^*,\eta_y^*$ and $\eta_B^*$ are discrete, i.e., made up of only isolated points, respectively.
\end{assum}

Then we have the following theorem.
\begin{thm}\label{thm10}
Assume that $P(\eta)$ is any parameterization of the kernel matrix
such that $P(\eta)$ is positive definite and moreover, $\Phi^T\Phi/N
\xra{} \Sigma$ almost surely  as $N\xra{}\infty$, where $\Sigma$ is
positive definite. Then we have as $N\xra{}\infty$
\begin{align}
&\widehat{\eta}_{\rm MSEg} \xra{}\eta_{\rm g}^*,~~
\widehat{\eta}_{\rm Sg} \xra{}\eta_{\rm g}^*, \label{ohpsg}\\
&\widehat{\eta}_{\rm MSEy} \xra{}\eta_{\rm y}^*,~~\widehat{\eta}_{\rm Sy} \xra{}\eta_{\rm y}^* \label{ohpsy},\\
&\widehat{\eta}_{{\rm EEB}}\xra{}\eta_{\rm
B}^*,~~\widehat{\eta}_{\rm EB} \xra{}\eta_{\rm B}^*, \label{ohpml}
\end{align}
almost surely.
Moreover, $\eta_{\rm g}^*$, $\eta_{\rm y}^*$, and $\eta_{\rm B}^*$ are a root of the system of equations, respectively, $i=1,\dots,p$:
    \begingroup
    \allowdisplaybreaks
    \begin{align*}
    &{\rm Tr}\Big( P(\eta)^{-1} \Sigma^{-2}
    P(\eta)^{-1}
    \big(P(\eta)  -
    \theta_0\theta_0^T\big) P(\eta)^{-1}
    \frac{\partial P(\eta)}{\partial \eta_i}\Big)\!=\!0,\\
    &{\rm Tr}\Big( P(\eta)^{-1} \Sigma^{-1}
    P(\eta)^{-1}
    \big(P(\eta)  -
    \theta_0\theta_0^T\big) P(\eta)^{-1}
    \frac{\partial P(\eta)}{\partial \eta_i}\Big)\!=\!0,\\
    &{\rm Tr}\Big( P(\eta)^{-1}
    \big(P(\eta)  -
    \theta_0\theta_0^T\big) P(\eta)^{-1}
    \frac{\partial P(\eta)}{\partial \eta_i}\Big)\!=\!0.
    \end{align*}
\endgroup
\end{thm}
The Oracle estimators $\widehat{\eta}_{\rm MSEg}$ and
$\widehat{\eta}_{\rm MSEg}$ are optimal for any data length $N$ in
the average sense if we are concerned with the ability to reproduce
the true impulse response and predict the future outputs of the
system respectively, while the SUREg $\widehat{\eta}_{\rm
Sg}$ and the SUREy $\widehat{\eta}_{\rm Sy}$ are not
optimal in general. Surprisingly, a nice property of
$\widehat{\eta}_{\rm Sg}$ and $\widehat{\eta}_{\rm Sy}$ is that they
converge to the best possible hyperparameter $\eta_{\rm g}^*$ and
$\eta_{\rm y}^*$, respectively, for any chosen parameterized kernel
matrix $P(\eta)$. It is so to speak that the two SURE methods are
``asymptotically consistent or asymptotically optimal''. This means
that when $N$ is sufficiently large, $\widehat{\eta}_{\rm Sg}$ and
$\widehat{\eta}_{\rm Sy}$ perform as well as $\widehat{\eta}_{\rm
MSEg}$ and $\widehat{\eta}_{\rm MSEy}$, respectively. It is also
worth noting that even with increasing number of data the EB
estimator $\widehat{\eta}_{\rm EB}$ has another preference than to
minimize MSEg and MSEy.
\begin{rem}
In contrast with $W_g(P,\Sigma,\theta_0)$ and
$W_y(P,\Sigma,\theta_0)$, a unique property of $W_{\rm
B}(P,\theta_0)$ is that it does not depend on the limit $\Sigma$ of
$\Phi^T\Phi/N$. This can to some extent explain why the EB estimator
is more robust  than the SUREg and SUREy, when
$\Phi^T\Phi$ is ill-conditioned. Interested readers can find
experimental evidence for this in \cite{Pillonetto2015}. However, in
contrast with the SUREg and SUREy, the EB estimator is
not asymptotically optimal.
    \end{rem}

    \begin{rem}
        The different expressions of the limit functions $W_g(P(\eta),\Sigma,\theta_0)$, $W_y(P(\eta),\Sigma,\theta_0)$, and $W_{\rm B}(P(\eta),\theta_0)$
        imply that the optimal hyperparameters $\eta_{\rm g}^*$, $\eta_{\rm y}^*$, and $\eta_{\rm B}^*$ may be different.
        To check this, we consider the ridge regression case, where $P=\eta I_n$ with $\eta>0$.
        In this case, \dref{osg}, \dref{osy} and \dref{oml} become
        \begingroup
        \allowdisplaybreaks
        \begin{align*}
        &\eta_{\rm g}^* =\arg\min_{\eta\geq 0}
        \frac{\sigma^4}{\eta^2}\theta_0^T \Sigma^{-2}\theta_0
        \!-\!\frac{2\sigma^4}{\eta}{\rm Tr}(\Sigma^{-2})
        =\frac{\theta_0^T \Sigma^{-2}\theta_0}{{\rm Tr}(\Sigma^{-2})},\\
        &\eta_{\rm y}^* =\arg\min_{\eta\geq 0}
        \frac{\sigma^4}{\eta^2}\theta_0^T \Sigma^{-1}\theta_0
        \!-\!\frac{2\sigma^4}{\eta}{\rm Tr}(\Sigma^{-1})
        =\frac{\theta_0^T \Sigma^{-1}\theta_0}{{\rm
        Tr}(\Sigma^{-1})},\\
        &\eta_{\rm B}^* =\arg\min_{\eta\geq 0} \theta_0^T\theta_0/\eta +\log \eta^n =  \theta_0^T\theta_0/n.
        \end{align*} which shows that $\eta_{\rm g}^*, \eta_{\rm y}^*$ and $\eta_{\rm
        B}^*$ can be different.
        \endgroup
        Clearly, when $\Sigma=dI_n$ with $d>0$, $\eta_{\rm g}^* = \eta_{\rm y}^* = \eta_{\rm B}^*$.
    \end{rem}

\begin{cor}
    \label{cor2}
    Assume that $\Phi^T\Phi/N \xra{} dI_n$ almost surely with $d>0$
    and $P(\eta)$ is any positive definite parameterization of the kernel matrix.
    Then we have
    \begingroup
    \allowdisplaybreaks
    \begin{align*}
    &\eta_{\rm g}^* = \eta_{\rm y}^*
    = \argmin_{\eta \in \Omega}
    \theta_0^TP(\eta)^{-2}\theta_0
    -2{\rm Tr} (P(\eta)^{-1}),\\
&\eta_{\rm B}^*\!=\!\argmin_{\eta \in \Omega}
\theta_0^TP(\eta)^{-1}\theta_0 +\log \det(P(\eta)).
    \end{align*}
    \endgroup
     and further $\eta_{\rm g}^*$ and  $\eta_{\rm B}^*$
    are roots of the following system of equations, respectively:
    \begingroup
    \allowdisplaybreaks
    \begin{align*}
&{\rm Tr} \Big(P(\eta)^{-2}\big(P(\eta) - \theta_0\theta_0^T \big)
P(\eta)^{-1} \frac{\partial P(\eta)}{\partial \eta_i}\Big)=0,
i=1,\dots,p,\\&
        {\rm Tr}
        \Big(P(\eta)^{-1}\big(P(\eta) - \theta_0\theta_0^T \big) P(\eta)^{-1}
        \frac{\partial P(\eta)}{\partial \eta_i}\Big)=0, i=1,\dots,p.
        \end{align*}
        \endgroup
In addition, for the diagonal kernel matrix \dref{dk}, we have
      \begingroup
      \allowdisplaybreaks
      \begin{align*}
      &\eta_{\rm g}^*= \eta_{\rm y}^*=\eta_{\rm B}^*=\left[(g_1^{0})^2,\cdots,(g_n^{0})^2\right]^T.
      \end{align*}
      \endgroup
\end{cor}
 In Theorem \ref{thm10}, we have considered the
convergence of those hyperparameter estimators.
In fact, we can
further derive their corresponding convergence rate. To this end, we
let $\xi_N=o_p(a_N)$ denote that the sequence
$\{\xi_N/a_N\}$ for nonzero sequence $\{a_N\}$ converges in
probability to zero, i.e., $\forall \epsilon>0,
P(|\xi_N/a_N|>\epsilon) \rightarrow0$ as
$N\rightarrow\infty$, while $\xi_N=O_p(a_N)$ denote that
$\{\xi_N/a_N\}$ is bounded in probability, i.e., $\forall
\epsilon>0,\exists L>0$ such that $P(|\xi_N/a_N|>L)<\epsilon,~\forall
N$. Then we have the following theorem.

\begin{figure*}[t]
    \centering
    \hspace*{-1em}
    \mbox{\includegraphics[scale=0.45]{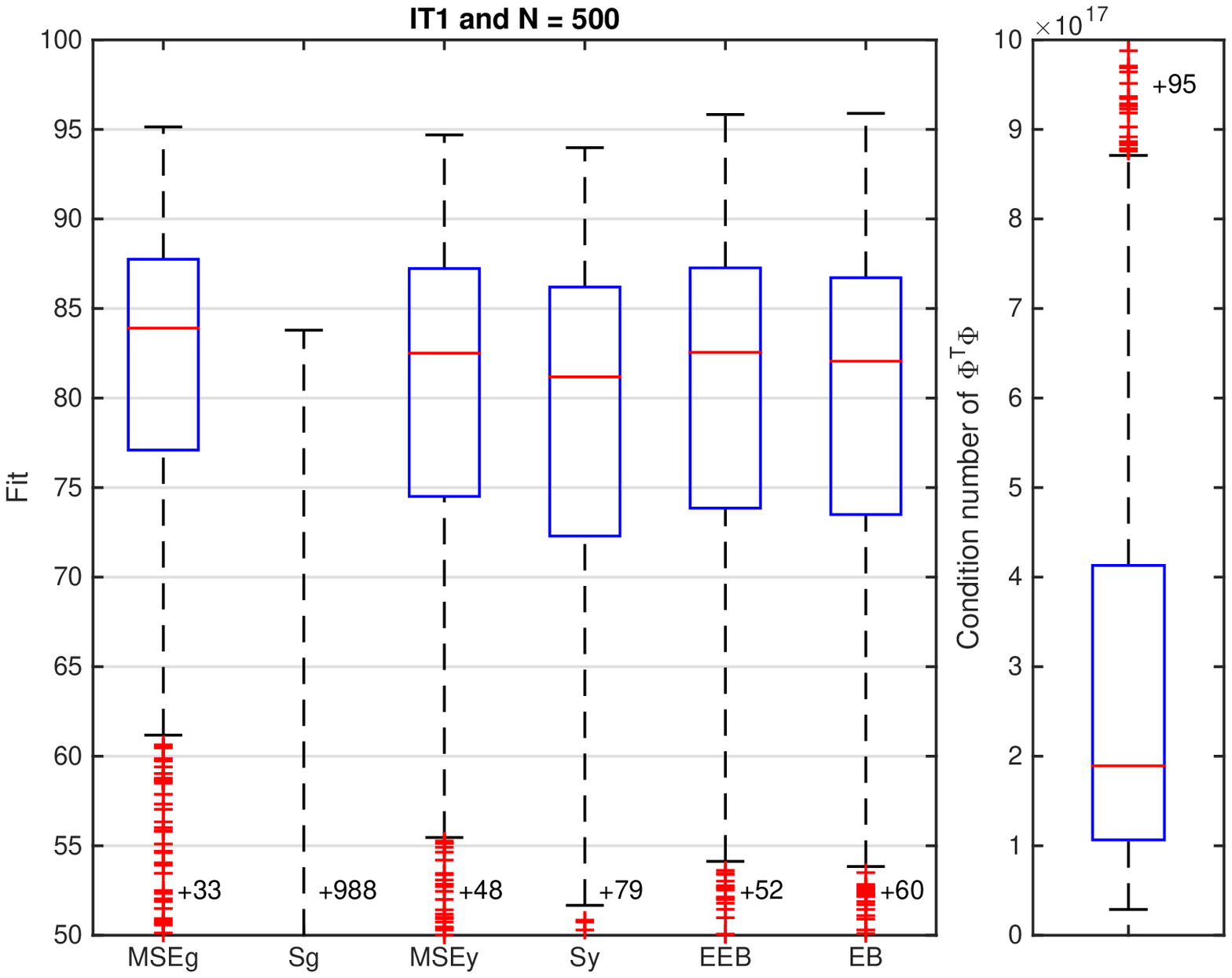}
    \includegraphics[scale=0.45]{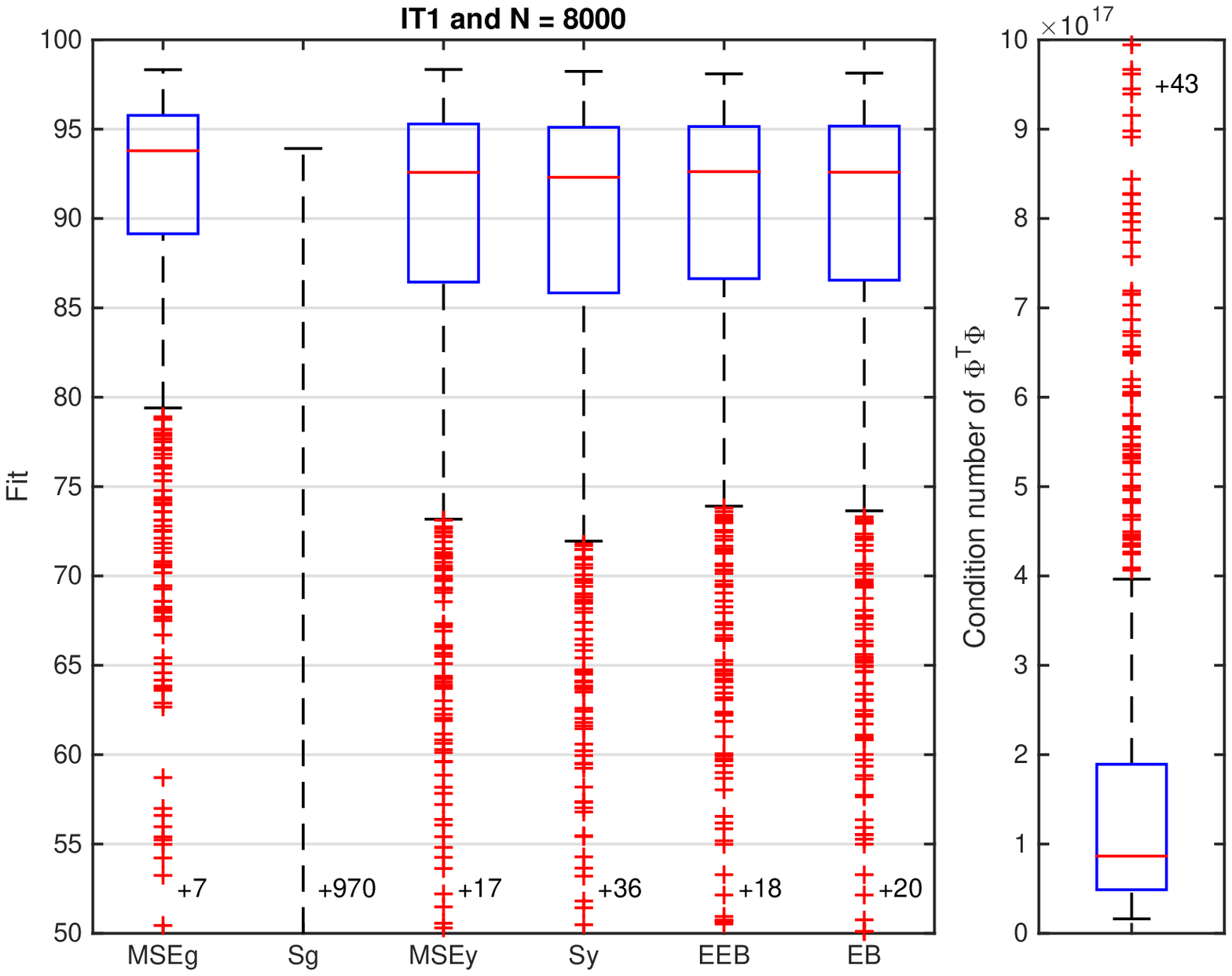}}
    \vspace{-3ex}
    \caption{Boxplot of the 1000 fits for the bandlimited white Gaussian noise input with the normalized band $[0,0.6]$ and boxplot of the condition numbers of the matrix
    $\Phi^T\Phi$: data lengths $N=500$ (left) and $N=8000$ (right).}
    \label{f1}
\end{figure*}

\begin{figure*}[t]
    \centering
    \hspace*{-1em}
    \mbox{\includegraphics[scale=0.45]{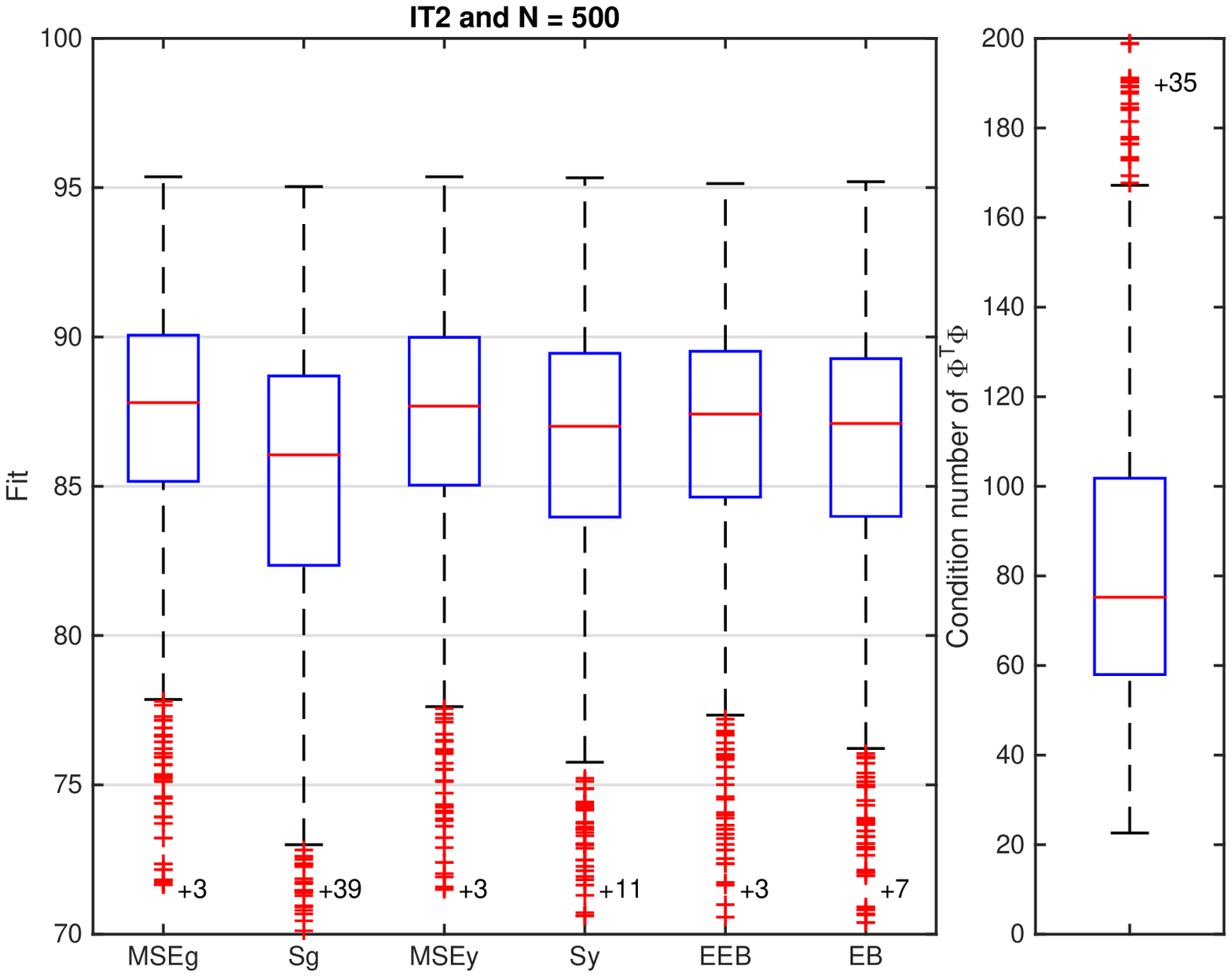}
    \includegraphics[scale=0.45]{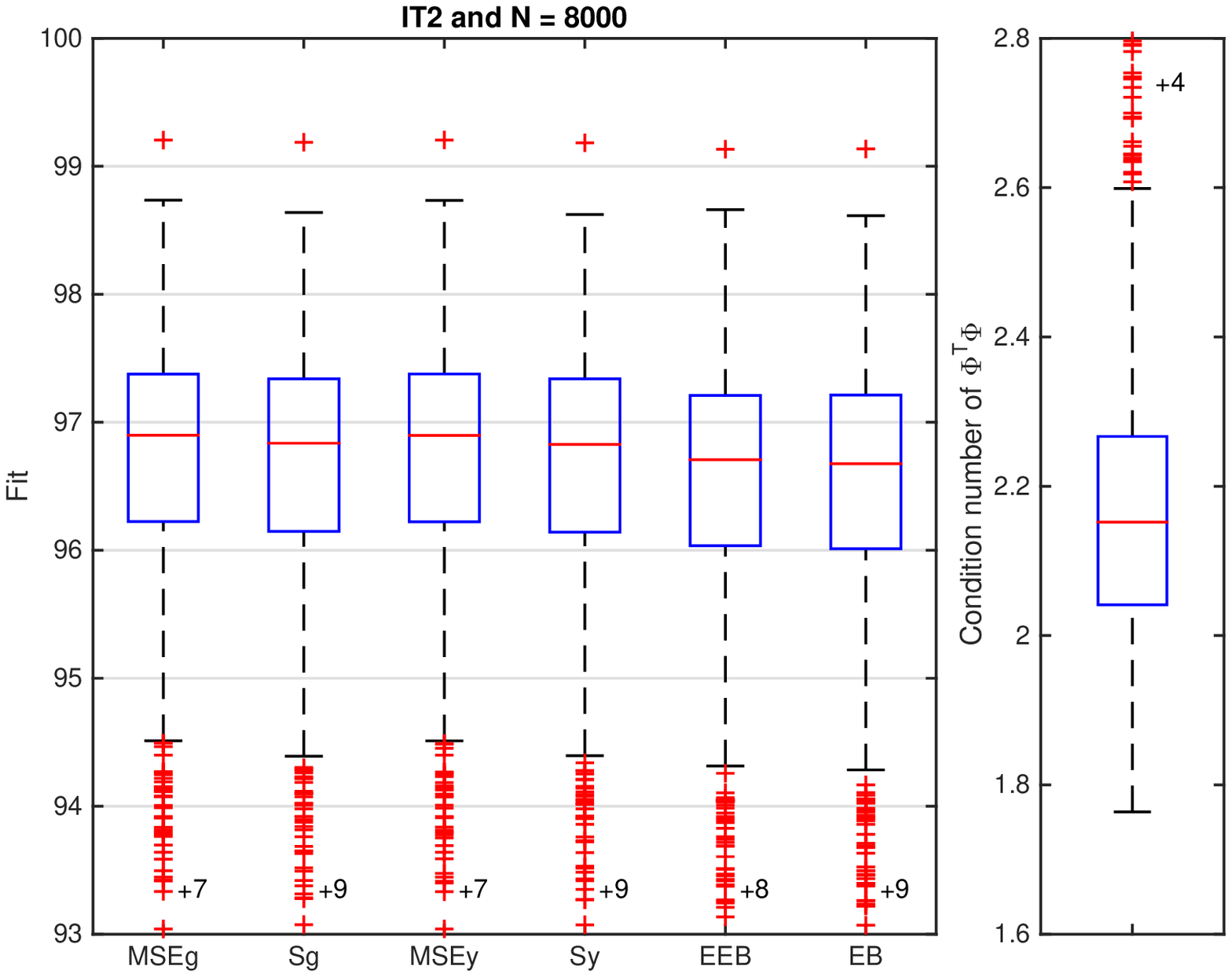}}
    \vspace{-3ex}
    \caption{Boxplot of the 1000 fits for the bandlimited white Gaussian noise input with the normalized band $[0,1]$ and boxplot of the condition numbers of the matrix
    $\Phi^T\Phi$: data lengths $N=500$ (left) and $N=8000$ (right).}    \label{f2}
\end{figure*}
\begin{thm}\label{thm12}
    Assume that $\|\Phi^T\Phi/N - \Sigma\|=O_p(\delta_N)$, where $\|\cdot\|$ denotes the Frobenius norm for a square matrix, $\delta_N\xra{} 0$ as $N\xra{}\infty$
    and $P(\eta)$ is any positive definite parameterization of the kernel matrix.
    Then we have
    \begingroup
    \allowdisplaybreaks
    \begin{align}
    &\|\widehat{\eta}_{\rm MSEg} -\eta_{\rm g}^*\|=O_p(\varpi_N),~
    \|\widehat{\eta}_{\rm Sg} - \eta_{\rm g}^*\|=O_p(\mu_N),\!\! \label{ohpsgr}\\
    &\|\widehat{\eta}_{\rm MSEy} -\eta_{\rm y}^*\|=O_p(\varpi_N),~\|\widehat{\eta}_{\rm Sy} - \eta_{\rm y}^*\|=O_p(\mu_N),\!\! \label{ohpsyr}\\
    &\|\widehat{\eta}_{{\rm EEB}} -\eta_{\rm B}^*\|=O_p(1/N),~~\|\widehat{\eta}_{\rm EB} -\eta_{\rm B}^*\|=O_p(1/\sqrt{N}),\!\!\!
    \label{ohpmlr}\\
    &\varpi_N\!=\!\max\big(O_p(\delta_N),O_p(1/N)\big),\label{eq:sequence1}\\
    &\mu_N\!=\!\max\big(O_p(\delta_N),O_p(1/\sqrt{N})\big)\label{eq:sequence2}.
    \end{align}
    \endgroup
    \end{thm}
Theorem \ref{thm12} shows that the convergence rate of
$\widehat{\eta}_{{\rm EEB}}$ and $\widehat{\eta}_{\rm EB}$ to
$\eta_{\rm B}^*$
    depends only on the fact $\Phi^T\Phi/N \xra{} \Sigma$ as $N\xra{}\infty$ ($\Phi^T\Phi=O_p(N)$) but not on the rate $\|\Phi^T\Phi/N - \Sigma\|=O_p(\delta_N)$.
    Moreover, we have
\begin{itemize}
\item the convergence rate of
    $\widehat{\eta}_{{\rm EEB}}$ to $\eta_{\rm B}^*$ is faster than
    that of $\widehat{\eta}_{\rm MSEg}$ to $\eta_{\rm g}^*$ and that
    of $\widehat{\eta}_{\rm MSEy}$ to $\eta_{\rm y}^*$.

\item the
    convergence rate of $\widehat{\eta}_{\rm EB}$ to $\eta_{\rm
    B}^*$ is faster than that of
    $\widehat{\eta}_{\rm Sg}$ to  $\eta_{\rm g}^*$ and that
    of $\widehat{\eta}_{\rm Sy}$ to $\eta_{\rm y}^*$.

\item the convergence rate of $\widehat{\eta}_{\rm
MSEg}$, $\widehat{\eta}_{\rm MSEy}$ and $\widehat{\eta}_{{\rm EEB}}$
to $\eta_{\rm g}^*$, $\eta_{\rm y}^*$ and $\eta_{\rm B}^*$,
respectively, is faster than that of $\widehat{\eta}_{\rm Sg}$,
$\widehat{\eta}_{\rm Sy}$ and $\widehat{\eta}_{{\rm EB}}$ to
$\eta_{\rm g}^*$, $\eta_{\rm y}^*$ and $\eta_{\rm B}^*$,
respectively.
\end{itemize}
Theorem \ref{thm12} has the following corollary.
\begin{cor}\label{cor5}
    Assume that $\|\Phi^T\Phi/N - \Sigma\|=O_p(\delta_N)$, where  $\delta_N\xra{} 0$ as $N\xra{}\infty$
    and $P(\eta)$ is any positive definite parameterization of the kernel matrix.
    Then
    \begingroup
    \allowdisplaybreaks
    \begin{align}
    &\|\widehat{\eta}_{\rm MSEg} -\widehat{\eta}_{\rm Sg}\|=O_p(\mu_N),\\
    &\|\widehat{\eta}_{\rm MSEy} -\widehat{\eta}_{\rm Sy} \|=O_p(\mu_N), \\
    &\|\widehat{\eta}_{{\rm EEB}} -\widehat{\eta}_{\rm EB}
    \|=O_p(1/\sqrt{N}),
    \end{align} where $\mu_N$ is defined in \eqref{eq:sequence2}.
    \end{cor}
    \endgroup
This corollary shows that the convergence rate of
$\|\widehat{\eta}_{{\rm EEB}} -\widehat{\eta}_{\rm EB} \|$ to zero
is faster than that of $\|\widehat{\eta}_{\rm MSEg}
-\widehat{\eta}_{\rm Sg}\|$ and $\|\widehat{\eta}_{\rm MSEy}
-\widehat{\eta}_{\rm Sy} \|$ to zero.

\section{Numerical Simulation}
\label{sec6} In this section, we illustrate the theoretical results
with numerical simulation.

\subsection{Test data-bank}
The method in \cite{Chen2012,Pillonetto2015} is used to generate
1000 30th order test systems. Then for each test system, we consider
four different test inputs:
\begin{itemize}
\item The first two test inputs are implemented by the MATLAB command {\tt idinput} choosing the bandlimited white Gaussian noise with
normalized bands $[0,0.6]$ and $[0,1]$, respectively, and denoted by
IT1 and IT2, respectively.

\item The third and fourth test inputs are the white Gaussian noise of unit variance filtered
by a second order rational transfer function $1/(1 - aq^{-1})^2$
with $a$ chosen to be $0.95$ and $0.05$, respectively, and denoted
by IT3 and IT4, respectively.

\end{itemize}

To generate the data set, we simulate each system with one of the
four test inputs to get the output, which is then corrupted by  an
additive white Gaussian noise. The signal-to-noise ratio  (SNR),
i.e., the ratio between the variance of the noise-free output and
the noise, is uniformly distributed over $[1,10]$, and is kept same
for the four test inputs.

Finally, in order to test the finite sample and asymptotic behavior
of the hyperparameter estimators, we consider data sets with
different data lengths $N=500$ and $8000$, respectively.

\subsection{Simulation Setup}

The performance of the RLS estimator \dref{rls} is evaluated by the
measure of fit \citep{Ljung2012} defined as follows :
\begin{align*} \mbox{\rm Fit} =100\times \left( 1 -
\frac{\|\widehat{\theta} -
\theta_0\|}{\|\theta_0-\bar{\theta}_0\|}\right),~~\bar{\theta}_0=\frac{1}{n}\sum_{k=1}^{n}g_k^0
\end{align*}
where $n$ is set to $200$. This fit is actually to evaluate the RLS
estimator in the MSEg sense.

The TC kernel \dref{tc} is considered and its hyperparameter
$\eta=[c,\alpha]^T$ is estimated by using the estimators SUREg
\dref{thsg}, SUREy \dref{thsy},  and EB \dref{thml}, respectively.
For reference, we also consider their corresponding Oracle
counterparts, i.e., the estimators MSEg \dref{okmseg}, MSEy
\dref{okmsey}, and EEB \dref{okmseml}, respectively. The notations
Sg, Sy, EB, MSEg, MSEy, and EEB are used to denote the corresponding
simulation results, respectively.

\subsection{Simulation results}
The average fits are given in Table \ref{tab1}.
\begin{table}[!ht]
    \centering
    \caption{Average fits for 1000 test systems and data sets.}
    \vspace{1.5ex}
    \begin{tabu} to 0.48
        \textwidth{X[1.3,c]X[0.7,c]X[0.7,c]X[0.7,c]X[0.6,c]X[0.6,c]X[1,c]}
        \hline
        &   MSEg    &   Sg  &     MSEy  &   Sy  &    EEB    &   EB      \\  \hline
    &&&IT1 \\ \hline
$N\!=\!500$     &   80.34   &   \mbox{-2.4E9}   &   78.07   &   53.83   &   77.98   &   77.26       \\
$N\!=\!8000$    &   90.63   &   \mbox{-8.6E8}   &   88.08   &
78.39   &   88.39   &   88.36       \\  \hline

    &&&IT2 \\ \hline
$N\!=\!500$     &   87.11   &   84.46   &    87.02  &   86.03   &   86.60   &   86.16       \\
$N\!=\!8000$    &   96.67   &   \mbox{96.60}    &    96.67  &
96.60   &   96.47   &   96.44       \\  \hline

    &&&IT3 \\ \hline
$N\!=\!500$     &   46.95   &   \mbox{-2220}    &    41.61  &   \mbox{-146.4}   &   39.47   &   39.03       \\
$N\!=\!8000$    &   57.67   &   \mbox{-176.8}   &    53.63  &
38.86   &   51.05   &   50.86       \\  \hline

    &&&IT4 \\ \hline
$N\!=\!500$     &   86.78   &   83.89   &    86.69  &   85.66   &   86.24   &   85.84       \\
$N\!=\!8000$    &   96.57   &   96.49   &    96.56  &   96.49   &
96.38   &   96.35       \\  \hline

    \end{tabu}
    \label{tab1}
\end{table}
The boxplots of the 1000 fits for IT1 and IT2 are displayed in Figs.
\ref{f1}--\ref{f2}, respectively. The boxplots for IT3 and IT4 are
skipped because of their similarity with IT1 and IT2.

\subsection{Findings}

Firstly, for all tested cases and in terms of average accuracy and
robustness, the Oracle estimators MSEg and MSEy (not implementable
in practice) are better than Sg and Sy, respectively, while EB is
just a little bit worse than but very close to its Oracle estimator
EEB.

Secondly, we consider the cases with input IT1, where $\Phi^T\Phi$
is very ill-conditioned for both $N=500$ and $N=8000$. In this case
and in terms of average accuracy and robustness, Sg performs badly
because it depends on $(\Phi^T\Phi)^{-1}$. Moreover, Sy is better
than Sg, but worse than EB.

Thirdly, we consider the case with input IT2 and $N=500$, where
$\Phi^T\Phi$ is much better conditioned than the cases with input
IT1. In this case and in terms of average accuracy and robustness,
Sg behaves much better in contrast with the cases with input IT1.
Moreover, EB and Sy are quite close though EB is a little bit
better, and they are all better than Sg.

Lastly, we consider the case with input IT2 and $N=8000$, where
$\Phi^T\Phi$ is very well-conditioned
and in terms of average accuracy and robustness, Sg behaves
much better in contrast with all the other cases, and performs as
well as Sy and better than EB. Moreover, Sg and Sy are very close to
the corresponding Oracle estimators MSEg and MSEy. These
observations coincide with the results found in Theorem \ref{thm10}
and Corollary \ref{cor2}. Namely, Sg and Sy are asymptotically
optimal but EB is not in the MSEg/MSEy senses and moreover, Sg and
Sy give the same optimal hyperparameter estimate as their Oracle
counterparts MSEg and MSEy, because the limit $\Sigma=I_n$ of
$\Phi^T\Phi/N$ as $N\rightarrow\infty$. It can also be seen from
Figs. \ref{f1} and \ref{f2} that the boxplots of EEB and EB is
closer than that of MSEg and Sg and that of MSEy and Sy. This
observation coincides with the result found in Corollary \ref{cor5},
that is, the convergence rate of $\|\widehat{\eta}_{{\rm EEB}}
-\widehat{\eta}_{\rm EB} \|$ to zero is faster than that of
$\|\widehat{\eta}_{\rm MSEg} -\widehat{\eta}_{\rm Sg}\|$ and
$\|\widehat{\eta}_{\rm MSEy} -\widehat{\eta}_{\rm Sy} \|$ to zero.

%

\section{Conclusions}
\label{sec7}

Kernel matrix design and hyperparamter estimation are two core
issues for the kernel based regularization methods. In contrast with
the former issue, there are few results reported for the latter
issue. In this paper, we focused on the latter issue and studied the
properties of several hyperparameter estimators including the
empirical Bayes (EB) estimator, two Stein's unbiased risk estimators
(SURE) and their corresponding Oracle counterparts, with an emphasis
on the asymptotic properties of these hyperparameter estimators. Our
major results are the following:
\begin{itemize}
\item The first order optimality conditions of these hyperparameter estimators are put in similar forms that better expose their relation and lead to several insights on these
    hyperparameter estimators.

  \item As the number of data
    goes to infinity, the two SUREs converge to the best hyperparameter
    minimizing the corresponding mean square error, respectively, while
    the more widely used EB estimator converges to another best
    hyperparameter minimizing the expectation of the EB estimation criterion. This indicates that
    the two SUREs are asymptotically optimal but the EB estimator is
    not.

  \item The convergence rate of two SUREs is slower than
    that of the EB estimator, and moreover, unlike the two SUREs, the EB
    estimator is independent of the convergence rate of $\Phi^T\Phi/N$ to
    its limit, where $\Phi$ is the regression matrix and $N$ is the
    number of data.
\end{itemize}
The results enhance our understanding about these hyperparameter
estimators and is one step forward towards the goal of building a
theory of the hyperparameter estimation for the kernel-based
regularization methods.

\appendix
\textbf{Appendix A}

Appendix A contains the proof of the results in the paper, for which
the technical lemmas are placed in Appendix B.
The proofs of Propositions \ref{prop1},  \ref{prop5}, \ref{prop6}, \ref{prop7},and \ref{prop8} and Corollaries \ref{cor4}, \ref{cor2}, and \ref{cor5}
are straightforward and thus omitted.

\renewcommand{\thesection}{A}

\subsection{Proof of Proposition \ref{th1}}

Under the setting $P^{-1}=\beta A/\sigma^2$, the MSEg
\dref{msege} of the RLS estimator \dref{rls} is a function of  $\beta$ for a given $A$:
\begin{align}
&\hspace{3em}{\rm MSEg}(\beta)={\rm Bias}(\beta)+{\rm Var}(\beta)~\mbox{where}\label{msegs}\\
\nonumber
&{\rm Bias}(\beta)=\beta^2\theta_0^{T}A^T(\Phi^T\Phi + \beta A)^{-1}(\Phi^T\Phi + \beta A)^{-1}A\theta_0,\\
\nonumber
&{\rm Var}(\beta)=\sigma^2{\rm Tr}\big((\Phi^T\Phi + \beta A)^{-1}\Phi^T\Phi(\Phi^T\Phi + \beta A)^{-1}\big).
\end{align}
Note that ${\rm MSEg}(0) = \sigma^2{\rm Tr}\big((\Phi^T\Phi)^{-1} \big)$ corresponds to the MSEg of the LS estimator \dref{ls}.
The derivatives of ${\rm Bias}(\beta)$ and ${\rm Var}(\beta)$ with respect to $\beta$ are as follows:
\begingroup
\allowdisplaybreaks
\begin{align}
\nonumber
\frac{\d {\rm Bias}(\beta)}{\d\beta}=&~2\beta\theta_0^TA^T(\Phi ^T\Phi  + \beta A)^{-1}(\Phi ^T\Phi  + \beta A)^{-1}A\theta_0\\
\nonumber
&-2\beta^2\theta_0^TA^T(\Phi ^T\Phi  + \beta A)^{-1}A(\Phi ^T\Phi  + \beta A)^{-1}\\
&~~~~\times(\Phi ^T\Phi  + \beta A)^{-1}A\theta_0 \label{db}\\
\nonumber
\frac{\d {\rm Var}(\beta)}{\d\beta}=&-2\sigma^2 {\rm Tr}
\big((\Phi ^T\Phi  + \beta A)^{-1}A(\Phi ^T\Phi  + \beta A)^{-1}\\
&~~~~\times\Phi ^T\Phi (\Phi ^T\Phi  + \beta A)^{-1}\big)
\end{align}
\endgroup
where the formula $\frac{\d C^{-1}(\beta)}{\d \beta}=-C^{-1}(\beta)\frac{\d C(\beta)}{\d \beta}C^{-1}(\beta)$
for an invertible matrix $C(\beta)$ is used.
Then we have
\begingroup
\allowdisplaybreaks
\begin{align}
\nonumber
&\frac{\d {\rm Bias}(\beta)}{\d\beta}\Big |_{\beta\xra{} 0^+}=0\\
\nonumber
&\frac{\d {\rm Var}(\beta)}{\d\beta}\Big |_{\beta\xra{} 0^+}=
-2\sigma^2 {\rm Tr}
\big((\Phi ^T\Phi )^{-1}A(\Phi ^T\Phi  )^{-1}\big)<0
\end{align}
\endgroup
where Lemma \ref{lm0} in Appendix B is used.
Therefore, we have
$
\frac{\d\rm MSEg(\beta)}{\d\beta}\Big |_{\beta\xra{} 0^+}<0.
$
This means that ${\rm MSEg}(\beta)< {\rm MSEg}(0)$ in some small right neighborhood of the origin $\beta =0$.

Under the assumption that $A$ is positive definite, denote
$$M(\beta)\eq E(\widehat{\theta}^{\rm R}-\theta_0)(\widehat{\theta}^{\rm R}-\theta_0)^T.$$
We first prove $M(0)-M(\beta)>0$ for $0<\beta<2\sigma^2/(\theta_0^T\!\!A\theta_0)$.
A straightforward calculation gives
\begin{align*}
&M(0)-M(\beta)\\
=&\sigma^2(\Phi^T\Phi )^{-1}
-\sigma^2(\Phi^T\Phi + \beta A)^{-1}\Phi^T\Phi(\Phi^T\Phi + \beta A)^{-1}\\
&-\beta^2(\Phi^T\Phi + \beta A)^{-1}A\theta_0\theta_0^TA(\Phi^T\Phi + \beta A)^{-1}\\
=&\beta(\Phi^T\Phi + \beta A)^{-1}\big(\sigma^2[2A+\beta A (\Phi^T\Phi)^{-1} A]
- \beta A\theta_0\theta_0^TA\big)\\
&~~~~\times(\Phi^T\Phi + \beta A)^{-1}.
\end{align*}
As a result, to prove $M(0)-M(\beta)>0$,
it suffices to show
\begin{align}
\sigma^2[2A+\beta A (\Phi^T\Phi)^{-1} A]
- \beta A\theta_0\theta_0^TA>0\label{pd}
\end{align}
which is true if $2\sigma^2I_n-\beta A^{1/2}\theta_0\theta_0^T A^{1/2}>0$
due to
\begin{align*}
\sigma^2&[2A+\beta A (\Phi^T\Phi)^{-1} A]
- \beta A\theta_0\theta_0^TA\\
&
>2\sigma^2A- \beta A\theta_0\theta_0^TA\\
&
=A^{1/2}(2\sigma^2I_n-\beta A^{1/2}\theta_0\theta_0^T A^{1/2})A^{1/2}>0.
\end{align*}
In addition, the eigenvalues of $A^{1/2}\theta_0\theta_0^T A^{1/2}$ are
$\theta_0^T\!A\theta_0$ and zero (with multiplicity $n\!\!-\!\!1$).
This shows
$2\sigma^2I_n-\beta A^{1/2}\theta_0\theta_0^T A^{1/2}>0$
for $0<\beta<2\sigma^2/(\theta_0^T\!A\theta_0)$.

Note that
${\rm MSEg}(\beta)={\rm Tr}(M(\beta))$.
One has proved that $M(0)-M(\beta)$ is positive definite if $0<\beta<2\sigma^2/(\theta_0^T\!A\theta_0)$,
so we have ${\rm MSEg}(0) -{\rm MSEg}(\beta)
\!=\!{\rm Tr}(M(0)-M(\beta))>0$.

The proof for the MSEy \dref{mseye}  is similar to that for the MSEg \dref{msege} by using the connection \dref{rgy}.
\begin{rem}\rm
        When $\beta\xra{}\infty$, from the MSEg \dref{msegs} we have
        \begingroup
        \allowdisplaybreaks
        \begin{enumerate}[1)]
            \item ${\rm Bias}(\beta)\xra{}\theta_0^T\theta_0$ and $\frac{{\rm d} {\rm Bias}(\beta)}{{\rm d} \beta}\xra{}0$,
            \item ${\rm Var}(\beta)\xra{}0$ and $\frac{{\rm d} {\rm Var}(\beta)}{{\rm d} \beta}\xra{}0$,
            \item ${\rm MSEg(\beta)}\xra{}\theta_0^T\theta_0$ and $\frac{\d {\rm MSEg(\beta)}}{\d\beta}\xra{}0$.
        \end{enumerate}
    \endgroup
\end{rem}

\subsection{Proof of Proposition \ref{thm3}}

We first prove \dref{dml}. Using the formulas \dref{md4} and
\dref{md2} derives that
\begingroup
\allowdisplaybreaks
\begin{align*}
\frac{\partial \mathscr{F}_{\rm EB}(P)}{\partial P}
=&\sum_{i,j} \big( - Q^{-T}YY^TQ^{-T}  + Q^{-T}\big)_{ij}
\frac{\partial  Q _{ij}}{\partial P}\\
=&-\Phi^TQ^{-T}YY^TQ^{-T}\Phi
+\Phi^TQ^{-T} \Phi.
\end{align*}
\endgroup
To prove \dref{dsg}, let us set
\begin{align*}
&\mathscr{F}_{\rm Sg_{1}}(P) = \sigma^4Y^TQ^{-T}\Phi (\Phi^T\Phi)^{-2}\Phi^T Q^{-1} Y\\
&\mathscr{F}_{\rm Sg_{2}}(P) = \sigma^2{\rm Tr}\big(2R^{-1}\!-(\Phi^T\Phi)^{-1}\big).
\end{align*}
By \dref{md1} and \dref{md3}, the derivative of $\mathscr{F}_{\rm S_{g1}}(P)$ is
\begingroup
\allowdisplaybreaks
\begin{align}
\nonumber
&\frac{\partial \mathscr{F}_{\rm Sg_{1}}(P)}{\partial P}
=\sigma^4\sum_{i,j} \big(2\Phi(\Phi^T\Phi)^{-2}\Phi^TQ^{-1}YY^T\big)_{ij}
\frac{\partial (Q^{-1} )_{ij}}{\partial P}\\
\nonumber
&=-2\sigma^4\sum_{i,j} (\Phi(\Phi^T\Phi)^{-2}\Phi^TQ^{-1}YY^T)_{ij}
\Phi^TQ^{-T}J_{ij}Q^{-T}\Phi\\
&=-2\sigma^4
\Phi^TQ^{-T} \Phi(\Phi^T\Phi)^{-2}\Phi^TQ^{-1}YY^TQ^{-T}\Phi. \label{dsg1}
\end{align}
\endgroup
and using \dref{id10} implies the derivative of $\mathscr{F}_{\rm S_{g2}}(P)$
\begingroup
\allowdisplaybreaks
\begin{align}
\nonumber
&\frac{\partial \mathscr{F}_{\rm Sg_{2}}(P)}{\partial P}
=2\sigma^2\sum_{i=1}^n
\frac{\partial (R^{-1})_{ii}}{\partial P}\\
=&~\!2\sigma^4
P^{-T}R^{-T}
R^{-T}P^{-T}
=2\sigma^4H^{-T}\overline{H}^{-T}.\label{dsg2}
\end{align}
\endgroup
Combining \dref{dsg1} with \dref{dsg2} derives \dref{dsg}.

Finally, let us prove \dref{dsy}.
Similarly, by using \dref{id9} we write \dref{sycf} as
\begin{align*}
\mathscr{F}_{\rm Sy}(P)
=&\sigma^4Y^TQ^{-T} Q^{-1} Y
 +(2\sigma^2 N-2\sigma^4{\rm Tr}( Q^{-1} ))\\
=& \mathscr{F}_{\rm Sy_{1}}(P) + \mathscr{F}_{\rm Sy_{2}}(P).
\end{align*}
By \dref{md1} and \dref{id7}, the derivative of $\mathscr{F}_{\rm S_{y1}}(P)$ is
\begingroup
\allowdisplaybreaks
\begin{align}
\nonumber
&\frac{\partial \mathscr{F}_{\rm S_{y1}}(P)}{\partial P}
=\sigma^4\sum_{i,j} \big(2Q^{-1}YY^T\big)_{ij}
\frac{\partial (Q^{-1} )_{ij}}{\partial P}\\
\nonumber
=&-2\sigma^4\sum_{i,j} \big(Q^{-1}YY^T\big)_{ij}
\Phi^T Q^{-T} J_{ij}
Q^{-T} \Phi\\
=&-2\sigma^4 \Phi^T Q^{-T} Q^{-1}YY^T Q^{-T} \Phi\label{dsy1}
\end{align}
\endgroup
and by using \dref{md5} the derivative of $\mathscr{F}_{\rm S_{y2}}(P)$ is
\begingroup
\allowdisplaybreaks
\begin{align}
\nonumber
\frac{\partial \mathscr{F}_{\rm S_{y2}}(P)}{\partial P}
&=-2\sigma^4\sum_{i=1}^n
\frac{\partial (Q^{-1} )_{ii}}{\partial P}\\
&=2\sigma^4\Phi^T Q^{-T} Q^{-T}\Phi.\label{dsy2}
\end{align}
\endgroup
The equations \dref{dsy1} and \dref{dsy2} implies \dref{dsy}.

\subsection{Proof of Proposition \ref{thmb1}:} It follows from
\dref{rls} that
\begingroup
\allowdisplaybreaks
\begin{align*}
\widehat{\theta}^{\rm R}\! -\!  \theta_0
&= R^{-1} \Phi^TY-\theta_0\\
&=-\sigma^2R^{-1}P^{-1}\theta_0
+R^{-1}\Phi^T V\\
&=-\sigma^2H^{-1}\theta_0
+R^{-1}\Phi^T V,
\end{align*}
\endgroup
which derives
\begin{align*}
{\rm MSEg}(P)
&=\sigma^4\theta_0^TH^{-T} H^{-1}\theta_0
+\sigma^2{\rm Tr}(R^{-1} \Phi^T \Phi R^{-T})\\
&={\rm MSEg1}(P) + {\rm MSEg2}(P).
\end{align*}
For the term ${\rm MSEg1}(P)$, using the formulas \dref{md1} and \dref{md3} gives
\begingroup
\allowdisplaybreaks
\begin{align}
\nonumber
&\frac{\partial {\rm MSEg1}(P)}{\partial P}
=\sigma^4\sum_{i,j}\big(2  H^{-1} \theta_0 \theta_0^T\big)_{ij}
\frac{\partial \big(H^{-1}\big)_{ij}}{\partial P}\\
\nonumber
=&\sigma^4\sum_{i,j}\big( 2H^{-1}  \theta_0 \theta_0^T\big)_{ij}
\big(-H^{-T}J_{ij}H^{-T}\Phi^T\Phi\big)\\
\nonumber
=&-2\sigma^4H^{-T}H^{-1}
\theta_0 \theta_0^T
H^{-T}\Phi^T\Phi\\
=&-2\sigma^4H^{-T}H^{-1}
\theta_0 \theta_0^T
\Phi^T
Q^{-T}\Phi.\label{dmseb1}
\end{align}
\endgroup
By using the formulas \dref{md6} and \dref{id10}, one derives
\begingroup
\allowdisplaybreaks
\begin{align}
\nonumber
&\frac{{\rm MSEg2}(P)}{\partial P}
=\sigma^2\sum_{i,j}
\big(2R^{-1}\Phi^T\Phi\big)_{ij}
 \frac{\partial\big( R^{-1} \big)_{ij}}{\partial P}\\
\nonumber
&=\sigma^2\sum_{i,j}
\big(2R^{-1}\Phi^T\Phi\big)_{ij}
(\sigma^2P^{-T} R^{-T} J_{ij} R^{-T}P^{-T})\\
\nonumber
&=2\sigma^4P^{-T} R^{-T}
R^{-1}\Phi^T\Phi R^{-T}P^{-T}\\
&=2\sigma^4
H^{-T}H^{-1}P\Phi^T
Q^{-T}\Phi.
\label{dmseb2}
\end{align}
\endgroup
Combining \dref{dmseb1} with \dref{dmseb2} implies the conclusion \dref{dmsege}.

In the following, we intend to prove \dref{dmseye}.
Let us set
\begin{align*}
{\rm MSE_{y1}}(P) &= \|\Phi P\Phi^TQ^{-1}\Phi \theta_0 - \Phi\theta_0\|^2 +N\sigma^2\\
&=\sigma^4\theta_0^T\Phi^{T}Q^{-T} Q^{-1}\Phi\theta_0 +N\sigma^2 \\
{\rm MSE_{y2}}(P)
&= \sigma^2{\rm Tr}\big(\Phi P \Phi^T Q^{-1} Q^{-T} \Phi P^T\Phi^T   \big)\\
&=\sigma^2{\rm Tr}\big((I_N \!-\! \sigma^2 Q^{-1})(I_N \!-\! \sigma^2 Q^{-T}) \big).
\end{align*}
By using \dref{md1} and \dref{id7},
one obtains
\begingroup
\allowdisplaybreaks
\begin{align}
\nonumber
&\frac{\partial {\rm MSE_{y1}}(P)}{\partial P}
=\sigma^4\sum_{i,j}\big(2  Q^{-1}\Phi \theta_0 \theta_0^T\Phi^T \big)_{ij}
\frac{\partial \big(Q^{-1}\big)_{ij}}{\partial P}\\
\nonumber
&\hspace{3em}=2\sigma^4 Q^{-1}\Phi \theta_0 \theta_0^T\Phi^T
(-\Phi^T Q^{-T} J_{ij} Q^{-T}\Phi)\\
&\hspace{3em}
=-2\sigma^4\Phi^T Q^{-T} Q^{-1}\Phi
\theta_0 \theta_0^T
 \Phi^T Q^{-T}\Phi.\label{dmsey1}
\end{align}
\endgroup
For the term ${\rm MSE_{y2}}(P)$,
using the formulas \dref{md6} and \dref{id7} derives
\begingroup
\allowdisplaybreaks
\begin{align}
\nonumber
&\frac{{\rm MSE_{y2}}(P)}{\partial P}
=\sigma^2\sum_{i,j}\big(2
\big(I_N - \sigma^2Q^{-1})\big)_{ij}
\frac{\partial \big(- \sigma^2Q^{-1} \big)_{ij}}{\partial P}\\
\nonumber
&=2\sigma^2\sum_{i,j}
\big(
I_N - \sigma^2Q^{-1}\big)_{ij}
(\sigma^2\Phi^TQ^{-T}J_{ij}
Q^{-T}\Phi)\\
\nonumber
&=2\sigma^4\Phi^TQ^{-T}
\big(I_N - \sigma^2Q^{-1}\big)
Q^{-T}\Phi\\
&=2\sigma^4
\Phi^T Q^{-T} Q^{-1}
\Phi P \Phi^T
Q^{-T}\Phi.\label{dmsey2}
\end{align}
\endgroup
Combining \dref{dmsey1} with \dref{dmsey2} implies the assertion \dref{dmseye}.

At last, we prove \dref{dmsemle}, which is derived by
\begin{align*}
\frac{\partial {\rm EEB}(P)}{\partial P}
&=\sum_{i,j}\big(- Q^{-T}\Phi \theta_0 \theta_0^T \Phi^T Q^{-T}
-\sigma^2Q^{-T}Q^{-T}\\
& \hspace{4.5em}
+ Q^{-T}\big)_{ij}
\frac{\partial Q_{ij}}{\partial P}\\
&\hspace{-4em}
=-\Phi^TQ^{-T}\Phi \theta_0 \theta_0^T \Phi^T Q^{-T} \Phi
+ \Phi^TQ^{-T}(I_N - \sigma^2 Q^{-T})  \Phi\\
&\hspace{-4em}
=-\Phi^TQ^{-T}\Phi \theta_0 \theta_0^T \Phi^T Q^{-T} \Phi
+\Phi^T Q^{-T} \Phi P^T\Phi^T Q^{-T} \Phi
\end{align*}
in terms of \dref{md4}, \dref{md2}, \dref{md5} and  \dref{id9}.

\subsection{Proof of Proposition \ref{thm11}}
Under the assumptions that $\Phi^T\Phi/N \xra{} \Sigma>0$ and the white noise $v(t)$,
we have $(\Phi^T\Phi)^{-1}=O_p(1/N)\xra{}0$, $S^{-1}\xra{}P^{-1}$, $NR^{-1}\xra{}\Sigma^{-1}$, $R^{-T}\Phi^T\Phi\xra{}I_n$, and $\widehat{\theta}^{\rm LS}\xra{}\theta_0$ almost surely as $N\xra{}\infty$.

Let us first prove \dref{cl1}. Using \eqref{eq:S},
we rewrite ${\rm MSEg}(P)$ in \dref{msege} as follows:
\begin{align*}
&{\rm MSEg}(P)=\sigma^4\theta_0^TS^{-T} (\Phi^T \Phi)^{-2}S^{-1}\theta_0\\
&\hspace{5.5em}+\sigma^2{\rm Tr}(R^{-1} \Phi^T \Phi R^{-T}).
\end{align*}
Noting ${\rm Tr}(\Sigma^{-1}P^{-1}\Sigma^{-1}) = {\rm
Tr}(\Sigma^{-1}P^{-T}\Sigma^{-1})$ and
\begingroup
\allowdisplaybreaks
\begin{align}
\nonumber
&N^2\big(R^{-1} \Phi^T \Phi R^{-T}
-(\Phi^T \Phi)^{-1}\big)\\
\nonumber
=&\!-\!\sigma^2\!N^2R^{-1}\big( P^{-1}\!+
\!P^{-T} \!\!+ \!\sigma^2 P^{-1} (\Phi^T \Phi)^{-1}\!P^{-T}
\big)R^{-T}\\
\xra{}&-\sigma^2\Sigma^{-1}(P^{-1}\!+\!P^{-T})\Sigma^{-1}\label{msegtr}
\end{align}
\endgroup
 yields that
\begingroup
\allowdisplaybreaks
\begin{align}
\nonumber
&N^2({\rm MSEg}(P) - \sigma^2{\rm Tr}( (\Phi^T \Phi)^{-1}))\\
\nonumber
=&~\sigma^4\theta_0^TS^{-T} (N^2(\Phi^T \Phi)^{-2})S^{-1}\theta_0\\
\nonumber &+\sigma^2N^2{\rm Tr}\big(R^{-1} \Phi^T \Phi R^{-T}
-(\Phi^T \Phi)^{-1}\big)\\
\nonumber \xra{}&\sigma^4\theta_0^TP^{-T}\Sigma^{-2}P^{-1}\theta_0
-2\sigma^4{\rm Tr} \big(
\Sigma^{-1}P^{-1}\Sigma^{-1}\big)\\
=& W_g(P,\Sigma,\theta_0).\label{msegwhole}
\end{align}
\endgroup
To prove \dref{cl2},
note that the first term of $\mathscr{F}_{\rm Sg}(P)$ can be rewritten as $\sigma^4(\widehat{\theta}^{\rm LS})^TS^{-T}(\Phi^T\Phi)^{-2}S^{-1}\widehat{\theta}^{\rm LS}$.
Thus one derives
\begin{align}
\nonumber
N^2(\mathscr{F}_{\rm Sg}(P)& - \sigma^2{\rm Tr}( (\Phi^T \Phi)^{-1}))\\
\nonumber
=&\sigma^4(\widehat{\theta}^{\rm LS})^TS^{-T}N^2(\Phi^T\Phi)^{-2}S^{-1}\widehat{\theta}^{\rm LS}\\
\nonumber
&+2\sigma^2N^2{\rm Tr}\big(R^{-1}\!-(\Phi^T\Phi)^{-1}\big)\\
\xra{}&W_g(P,\Sigma,\theta_0)\label{ga5}
\end{align}
where we use the limit
\begin{align*}
N^2(R^{-1} \!- (\Phi^T \Phi)^{-1})
&=-\sigma^2NR^{-1}P^{-1}N(\Phi^T \Phi)^{-1}\\
&\xra{}-\sigma^2\Sigma^{-1}P^{-1}\Sigma^{-1}.
\end{align*}
Similarly, we can rewrite ${\rm MSEy}(P)$ as
\begin{align}
\nonumber
{\rm MSEy}(P)
&=\sigma^4\theta_0^T S^{-T}(\Phi^T \Phi)^{-1}S^{-1}\theta_0 +N\sigma^2\\
&~~~~+{\rm Tr}\big(R^{-1}\Phi^T\Phi R^{-T}\Phi^T\Phi \big)
\end{align}
and hence the assertion \dref{cl3} is proved by
\begin{align}
\nonumber
N({\rm MSEy}&(P)\! -\! (n\!+\!N)\sigma^2)\\
\nonumber
=&\sigma^4\theta_0^T S^{-T}N(\Phi^T \Phi)^{-1}S^{-1}\theta_0\\
&+\sigma^2N{\rm Tr}\big( R^{-1}\Phi^T\Phi R^{-T}\Phi^T\Phi \!- I_n\big) \label{mmsey}\\
\nonumber
\xra{}&W_y(P,\Sigma,\theta_0)
\end{align}
where we use the formulas
\begin{align*}
&N(R^{-1}\Phi^T\Phi R^{-T}\Phi^T\Phi-I_n)\\
=&\!-\!\sigma^2NR^{-1}
\big[P^{-1} \!+\! P^{-T}
\!+\! \sigma^2P^{-1} (\Phi^T \Phi)^{-1}P^{-T}\big]R^{-T}\Phi^T\Phi
\\
\xra{}&-\sigma^2\Sigma^{-1}( P^{-1} + P^{-T})
\end{align*}
and ${\rm Tr}(\Sigma^{-1}P^{-1}) = {\rm Tr}(P^{-T}\Sigma^{-1}) = {\rm Tr}(\Sigma^{-1}P^{-T})$.

To prove \dref{cl4}, we need some identities.
A straightforward calculation shows that
\begin{align*}
Q^{T}(I_N - \Phi(\Phi^T \Phi)^{-1}\Phi^T )Q
=\sigma^4(I_N - \Phi(\Phi^T \Phi)^{-1}\Phi^T ).
\end{align*}
This means that
\begin{align*}
\sigma^4Q^{-T}(I_N - \Phi(\Phi^T \Phi)^{-1}\Phi^T )Q^{-1}
\!= I_N \!- \Phi(\Phi^T \Phi)^{-1}\Phi^T
\end{align*}
and hence we derive
\begin{align*}
\sigma^4Y^TQ^{-T}Q^{-1}Y
+Y^T\Phi(\Phi^T \Phi)^{-1}\Phi^TY
-Y^TY\\
=\sigma^4Y^TQ^{-T}\Phi(\Phi^T \Phi)^{-1}\Phi^TQ^{-1}Y.
\end{align*}
It follows from \dref{id6} and \dref{id9} that
\begin{align}
\nonumber
&N\big[\mathscr{F}_{\rm Sy}(P)
+ Y^T\Phi(\Phi^T\Phi)^{-1}\Phi^TY
-Y^TY
- 2n\sigma^2\big]\\
\nonumber
=&N\big[\sigma^4Y^TQ^{-T}\Phi(\Phi^T \Phi)^{-1}\Phi^TQ^{-1}Y
\!\!+\! 2\sigma^2 {\rm Tr}( R^{-1}\Phi^T \Phi \!- \!I_n)\big]\\
\nonumber
=&N\big[\sigma^4(\widehat{\theta}^{\rm LS})^TS^{-T}
\!(\Phi^T \Phi)^{-1}
S^{-1}\widehat{\theta}^{\rm LS}
\!\!+ \!2\sigma^2 {\rm Tr}( R^{-1}\Phi^T \Phi \!- \!I_n)\big]\\
\xra{}&W_y(P,\Sigma,\theta_0)\label{msy}
\end{align}
where we use the limit
\begin{align*}
N(R^{-1}\Phi^T \Phi- I_n)
=-\sigma^2NR^{-1}P^{-1}
\xra{}-\sigma^2\Sigma^{-1}P^{-1}.
\end{align*}
Similarly, we need two identities to prove \dref{cl5}.
Using the Sylvester's determinant identity
$\det(I_n + AB) = \det(I_N + BA)$
derives
\begin{align*}
\det(Q)=\sigma^{2(N-n)}\det(\Phi^T\Phi)\det(P+\sigma^2(\Phi^T\Phi)^{-1})
\end{align*}
which implies
\begin{align}
\nonumber
\log\det (Q)
-(N-n)&\log \sigma^2 - \log \det(\Phi^T\Phi)\\
=&\log \det(S)
\xra{}\log \det(P).\label{eqdet}
\end{align}
Starting with the identity
$I_N\! =\! \sigma^2 Q^{-1} + \Phi P \Phi^T Q^{-1}$
gives
\begin{align*}
\sigma^2{\rm Tr}( Q^{-1}) &= N - {\rm Tr}( \Phi P \Phi^T Q^{-1}) \\
&=N-{\rm Tr}( R^{-1}\Phi^T \Phi)\xra{}N-n.
\end{align*}
Therefore, the limit \dref{cl5} is proved by
\begin{align}
\nonumber
&{\rm EEB}(P) \!-\! (N\!-\!n)
 \!-\! (N\!-\!n)\log \sigma^2 \!\! - \! \log \det(\Phi^T\Phi)\\
 \nonumber
 &=\theta_0^TS^{-1} \theta_0
 \!+ \!\big(\sigma^2{\rm Tr}( Q^{-1} ) - (N-n)  \big) \\
 \nonumber
 &~~+ \log\det (Q)-(N-n) \log \sigma^2 - \log \det(\Phi^T\Phi)\\
 &\xra{}\theta_0^TP^{-1} \theta_0 + \log \det(P)
 =W_{\rm B}(P,\theta_0).\label{mseml}
\end{align}
At last, we finish the proof  by checking \dref{cl6}.
The identity $$Q(I_N - \Phi(\Phi^T \Phi)^{-1}\Phi^T )/\sigma^2
=I_N - \Phi(\Phi^T \Phi)^{-1}\Phi^T$$
implies that
\begin{align}
\nonumber
Y^TQ^{-1}Y +Y^T\Phi(\Phi^T \Phi)^{-1}\Phi^TY/\sigma^2
-Y^TY/\sigma^2\\
= Y^TQ^{-1}\Phi(\Phi^T \Phi)^{-1}\Phi^T Y. \label{eqyy}
\end{align}
It follows from \dref{eqdet}, \dref{eqyy}, and \dref{id6}   that
    \begin{align}
    \nonumber
    &\mathscr{F}_{\rm EB}(P)
    + Y^T\Phi(\Phi^T\Phi)^{-1}\Phi^TY/\sigma^2
    -Y^TY/\sigma^2\\
    \nonumber
    &\hspace{0.4cm} - (N\!-\!n)\log \sigma^2 \! - \! \log \det(\Phi^T\Phi)\\
    \nonumber
     =& Y^TQ^{-1}Y +Y^T\Phi(\Phi^T \Phi)^{-1}\Phi^TY/\sigma^2
    -Y^TY/\sigma^2\\
    \nonumber
    &+\log\det (Q)-(N-n) \log \sigma^2 - \log \det(\Phi^T\Phi)\\
    \nonumber
    =&Y^TQ^{-1}\Phi(\Phi^T \Phi)^{-1}\Phi^T Y
    +\log \det(S)\\
    \xra{}& W_{\rm B}(P,\theta_0).\label{mml}
    \end{align}

\subsection{Proof of Theorem \ref{thm10}}
Firstly, we prove $\widehat{\eta}_{\rm MSEg} \xra{}\eta_{\rm g}^*$
as $N\xra{}\infty$. Define
 \begin{align}
 \overline{{\rm MSE}}{\rm g}(P)
 \eq N^2\big({\rm MSEg}(P) - \sigma^2{\rm Tr}( (\Phi^T \Phi)^{-1})\big).\!\!\label{mmseg}
 \end{align}
Clearly, we have $\widehat{\eta}_{\rm MSEg}$ also minimizes
$\overline{{\rm MSE}}{\rm g}(P(\eta))$, i.e.,
\begin{align*}
\widehat{\eta}_{\rm MSEg} = \argmin_{\eta \in \Omega} \overline{{\rm
MSE}}{\rm g}(P(\eta)).
\end{align*}
Under Assumption \ref{ass:1}, there exists a compact set
\begin{align}\label{eq:barOmega}
\overline{\Omega}\subset \Omega\end{align}  containing $\eta_{\rm
g}^*$ such that $0<d_1\leq \|P(\eta)\|\leq d_2<\infty$ for all
$\eta\in \overline{\Omega}$. Then by Lemma \ref{ct} in Appendix B,
to prove $\widehat{\eta}_{\rm MSEg} \xra{}\eta_{\rm g}^*$ as
$N\xra{}\infty$, it suffices to show that $\overline{{\rm MSE}}{\rm
g}(P(\eta)) $ converges to $W_g(P(\eta),\Sigma,\theta_0)$ almost
surely and uniformly in $\overline{\Omega}$, as
$N\rightarrow\infty$.

It follows from \dref{msegwhole} and \dref{msegtr} that
\begingroup
\allowdisplaybreaks
 \begin{align}
 \nonumber
\overline{{\rm MSE}}{\rm g}&(P(\eta)) - W_g(P,\Sigma,\theta_0)\\
=&\sigma^4Z_1 + 2\sigma^4{\rm Tr}
\big(Z_2\big)
-\sigma^6{\rm Tr}\big(Z_3 \big),\\
 \nonumber
 Z_1=&~\theta_0^TS^{-T} (N^2(\Phi^T \Phi)^{-2})S^{-1}\theta_0
 -\theta_0^TP^{-T}\Sigma^{-2}P^{-1}\theta_0\\
 Z_2=&~\Sigma^{-1}P^{-1}\Sigma^{-1}
 -N^2R^{-1}P^{-1}R^{-T}\\
  \nonumber
 Z_3=&-N^2R^{-1}P^{-1} (\Phi^T \Phi)^{-1} P^{-T}R^{-T}.
 \end{align}
 \endgroup
For the term $Z_1$, we have
 \begin{align}
 \nonumber
Z_1 =\!~&\theta_0^T(S^{-T} - P^{-T}) (N^2(\Phi^T \Phi)^{-2})S^{-1}\theta_0\\
 \nonumber
 &+\theta_0^TP^{-T} (N^2(\Phi^T \Phi)^{-2} - \Sigma^{-2})S^{-1}\theta_0\\
 &+\theta_0^TP^{-T}\Sigma^{-2}(S^{-1}-P^{-1})\theta_0\label{ga1}
 \end{align}
where
\begin{align}
S^{-1}-P^{-1}
=-\sigma^2S^{-1}(\Phi^T\Phi)^{-1}P^{-1}.\label{ga2}
\end{align}
Note that $\Phi^T\Phi/N \xra{} \Sigma$ implies that
$\|N(\Phi^T\Phi)^{-1}\|=O_p(1)$. Then further noting that $d_1\leq
\|P(\eta)\|\leq d_2 $ and $\|S(\eta)^{-1}\| < \|(P(\eta))^{-1}\|\leq
1/d_1 $ for $\eta\in\overline{\Omega}$, we have $Z_1$ converges to
zero almost surely and uniformly in $\overline{\Omega}$. For the
term $Z_2$, we have
\begin{align}
\nonumber
 &\Sigma^{-1}P^{-1}\Sigma^{-1}
-N^2R^{-1}P^{-1}R^{-T}\\
\nonumber =&(\Sigma^{-1} \!-\!NR^{-1})P^{-1}\Sigma^{-1}
\!+\!NR^{-1}P^{-1}(\Sigma^{-1} \!-\! NR^{-T} ).
\end{align}
Noting $NR^{-1} \xra{} \Sigma^{-1}$ and
$\|NR^{-1}-\Sigma^{-1}\|=O_p(1)$ yields that $Z_2$ converges to zero
almost surely and uniformly in $\overline{\Omega}$. Finally, by
noting $(\Phi^T\Phi)^{-1}\xra{}0$ as $N\rightarrow\infty$ it is easy
to see that $Z_3$ also converges to zero almost surely and
uniformly. Making use of these facts shows that $\overline{{\rm
        MSE}}{\rm g}(P(\eta)) $ converges to $W_g(P(\eta),\Sigma,\theta_0)$
almost surely and uniformly in $\overline{\Omega}$ and hence, by
Lemma \ref{ct}, $\widehat{\eta}_{\rm MSEg} \xra{}\eta_{\rm g}^*$
as $N\rightarrow\infty$ almost surely.

Secondly, we prove that $\widehat{\eta}_{\rm Sg} \xra{}\eta_{\rm
    g}^*$ as $N\rightarrow\infty$ and the proof is similar to that of
$\widehat{\eta}_{\rm MSEg} \xra{}\eta_{\rm g}^*$ as
$N\rightarrow\infty$. Define
\begin{align*}
&\mathscr{\overline{F}}_{\rm Sg}(P(\eta))
\eq N^2\big(\mathscr{F}_{\rm Sg}(P(\eta)) - \sigma^2{\rm Tr}( (\Phi^T \Phi)^{-1})\big).
\end{align*}
Then, we have
\begin{align}
\widehat{\eta}_{\rm Sg}=\argmin_{\eta\in\Omega} \mathscr{\overline{F}}_{\rm Sg}(P(\eta)).
\end{align}
It follows from \dref{ga5} that
\begingroup
\allowdisplaybreaks
\begin{align}
\nonumber &\mathscr{\overline{F}}_{\rm Sg}(P(\eta))
-W_g(P,\Sigma,\theta_0)
=\sigma^4Z'_1+2\sigma^4{\rm Tr}\big(Z'_2 \big),\\
\nonumber
&Z'_1\!=\!(\widehat{\theta}^{\rm LS})^TS^{-T}\!N^2(\Phi^T\Phi)^{-2}S^{-1}\widehat{\theta}^{\rm LS}
\!-\!\theta_0^TP^{-T}\Sigma^{-2}P^{-1}\theta_0\\
\nonumber
&Z'_2\!=\!\Sigma^{-1}P^{-1}\Sigma^{-1}-NR^{-1}P^{-1}N(\Phi^T \Phi)^{-1}.
\end{align}
\endgroup
For the terms $Z'_1$ and $Z'_2$, we have
\begingroup
\allowdisplaybreaks
\begin{align}
\nonumber
Z'_1=&\big(\widehat{\theta}^{\rm LS}-\theta_0\big)^T
S^{-T}\!N^2(\Phi^T\Phi)^{-2}S^{-1}\widehat{\theta}^{\rm LS}\\
\nonumber
&+\theta_0^T\big(S^{-T} \!- P^{-T}\big)N^2(\Phi^T\Phi)^{-2}S^{-1}\widehat{\theta}^{\rm LS}\\
\nonumber
&+\theta_0^TP^{-T}\big(N^2(\Phi^T\Phi)^{-2} -\Sigma^{-2} \big)S^{-1}\widehat{\theta}^{\rm LS}\\
\nonumber
&+\theta_0^TP^{-T}\Sigma^{-2}\big(S^{-1}-P^{-1} \big)\widehat{\theta}^{\rm LS}\\
&+\theta_0^TP^{-T}\Sigma^{-2}P^{-1} \big( \widehat{\theta}^{\rm LS}- \theta_0\big)\label{z1}\\
\nonumber
Z'_2=&\big(\Sigma^{-1} - NR^{-1} \big)P^{-1}\Sigma^{-1}\\
&+NR^{-1} P^{-1}\big( \Sigma^{-1} - N(\Phi^T \Phi)^{-1} \big).\label{z2}
\end{align}
\endgroup
Then, noting that $\widehat{\theta}^{\rm
    LS}\xra{}\theta_0$, $S^{-1}\xra{}P^{-1}$,
$N(\Phi^T \Phi)^{-1}\xra{}\Sigma^{-1}$, $NR^{-1}\xra{}\Sigma^{-1}$
almost surely as $N\xra{}\infty$, and $ \|NR^{-1}\|=O_p(1)$,
$\|\widehat{\theta}^{\rm LS}\|=O_p(1)$, and $d_1\leq \|P(\eta)\|\leq
d_2 $, $\|S(\eta)^{-1}\| < \|(P(\eta))^{-1}\|\leq 1/d_1 $, for
$\eta\in\overline{\Omega}$, one can show that each term of \dref{z1}
and \dref{z2}, and thus both $Z'_1$ and $Z'_2$ converge to zero
almost surely and uniformly in $\overline{\Omega}$. Therefore,
$\mathscr{\overline{F}}_{\rm Sg}(P(\eta))$ converges to
$W_g(P,\Sigma,\theta_0)$ almost surely and uniformly in
$\overline{\Omega}$. It then follows from Lemma \ref{ct} that
$\widehat{\eta}_{\rm Sg} \xra{}\eta_{\rm g}^*$ almost surely as
$N\xra{}\infty$.

The proof of \dref{ohpsy} and \dref{ohpml} can be done similarly
and thus is omitted. The first order optimality conditions of
$\eta_g^*,\eta_y^*,$ and $\eta_{\rm B}^*$ can be derived in a
similar way as Proposition \ref{thmb1} and thus is omitted. This
completes the proof.

\subsection{Proof of Theorem \ref{thm12}}

We first prove that $\|\widehat{\eta}_{\rm MSEg} -\eta_{\rm
g}^*\|=O_p(\varpi_N)$.

Noting \dref{msegwhole}, the $i$-th elements of the gradient vectors
of $\overline{{\rm MSE}}{\rm g}(P(\eta))$ and $
W_g(P(\eta),\Sigma,\theta_0)$ with respect to $\eta$ are,
respectively, for $1\leq i \leq p$,
\begingroup
\allowdisplaybreaks
\begin{align}
\nonumber
&\frac{\partial \overline{{\rm MSE}}{\rm g}(P(\eta))}{\partial \eta_i}
\!=2\sigma^4N^2\theta_0^TS^{-T}( \Phi^T \Phi)^{-2}
\frac{\partial S^{-1}}{\partial \eta_i} \theta_0\\
\nonumber
&\hspace{7em}+2\sigma^2N^2{\rm Tr}\Big(\!  \frac{\partial R^{-1}}{\partial \eta_i}  \Phi^T \Phi R^{-T}\Big)\\
\nonumber
&\frac{\partial W_g(P(\eta),\Sigma,\theta_0)}{\partial \eta_i}
=2\sigma^4\theta_0^TP^{-T} \Sigma^{-2}
 \frac{\partial P^{-1}}{\partial \eta_i} \theta_0\\
&\hspace{8.5em}-2\sigma^4{\rm Tr}\Big( \Sigma^{-1} \frac{\partial P^{-1}}{\partial \eta_i}  \Sigma^{-1}\Big).\!\!\!\!
\label{gwg}
\end{align}
\endgroup
Using the identity $\frac{\partial R^{-1}}{\partial \eta_i}\! =\! -R^{-1}\frac{\partial R}{\partial \eta_i}R^{-1}\!=\!-\sigma^2R^{-1}\frac{\partial P^{-1}}{\partial \eta_i}R^{-1}$, we see their difference is
\begingroup
\allowdisplaybreaks
\begin{align}
\nonumber
&\frac{\partial \overline{{\rm MSE}}{\rm g}(P(\eta))}{\partial \eta_i}
\!-\!\frac{\partial W_g(P(\eta),\Sigma,\theta_0)}{\partial \eta_i}
=2\sigma^4 \big(\Upsilon_1 \!+\!{\rm Tr}(\Upsilon_2  )  \big),\\
\nonumber
&\mbox{where}~\Upsilon_1=
\theta_0^TS^{-T}\big(N^2( \Phi^T \Phi)^{-2}\big)
\frac{\partial S^{-1}}{\partial \eta_i} \theta_0\\
\nonumber
&\hspace{5.3em}-\theta_0^TP^{-T} \Sigma^{-2}
\frac{\partial P^{-1}}{\partial \eta_i} \theta_0,\\
\nonumber
&\Upsilon_2\!=\Sigma^{-1} \frac{\partial P^{-1}}{\partial \eta_i}  \Sigma^{-1}
-NR^{-1} \frac{\partial P^{-1}}{\partial \eta_i} R^{-1}  \Phi^T \Phi NR^{-T}.
\end{align}
\endgroup
Noting $\|N(\Phi^T\Phi)^{-1} - \Sigma^{-1}\|=O_p(\delta_N)$,
$\|S^{-1} - P^{-1}\| =O_p(1/N)$, $\big\|\frac{\partial
S^{-1}}{\partial \eta_i}
 - \frac{\partial P^{-1}}{\partial \eta_i}\big\|=O_p(1/N) $,
$\big\|R^{-1}  \Phi^T \Phi  \big\|=O_p(1/N)$, $\big\|NR^{-1} -
\Sigma^{-1}  \big\|=O_p(\delta_N)$, and $d_1\leq \|P(\eta)\|\leq d_2
$ and $\|S(\eta)^{-1}\| < \|(P(\eta))^{-1}\|\leq 1/d_1 $ for
$\eta\in\overline{\Omega}$ yields
\begin{align}
&|\Upsilon_1 | = O_p(\varpi_N),~~
|{\rm Tr}(\Upsilon_2  ) | = O_p(\varpi_N)
\end{align}
uniformly in $\overline{\Omega}$, where $\overline{\Omega}$ is
defined in \eqref{eq:barOmega}. Therefore, we have
\begin{align*}
\Big\|\frac{\partial \overline{{\rm MSE}}{\rm g}(P(\eta))}{\partial \eta}
-\frac{\partial W_g(P(\eta),\Sigma,\theta_0)}{\partial \eta}\Big\|
=O_p(\varpi_N)
\end{align*}
uniformly for any $\eta \in\overline{\Omega}$. Since
$\widehat{\eta}_{\rm MSEg}$ and  $\eta_{\rm g}^*$ minimize
$\overline{{\rm MSE}}{\rm g}(P)$ and $W_g(P,\Sigma,\theta_0)$,
respectively, we have
\begin{align*}
\frac{\partial \overline{{\rm MSE}}{\rm g}(P(\eta))}{\partial \eta}\Big|_{\eta = \widehat{\eta}_{\rm MSEg}}\!\! =0
~\mbox{and}~
\frac{\partial W_g(P(\eta),\Sigma,\theta_0)}{\partial \eta}\Big|_{\eta = \eta_{\rm g}^*} \!\!=0.
\end{align*}
It follows that
\begin{align*}
\frac{\partial \overline{{\rm MSE}}{\rm g}(P(\eta))}{\partial \eta}\Big|_{\eta = \eta_{\rm g}^*} =O_p(\varpi_N).
\end{align*}
In addition, by using \dref{gwg}, the $(i,j)$-element of the Hessian matrix of $W_g(P(\eta),\Sigma,\theta_0)$ is
\begingroup
\allowdisplaybreaks
\begin{align}
\nonumber
&\frac{\partial^2 W_g(P(\eta),\Sigma,\theta_0)}{\partial \eta_i \partial \eta_j}\\
\nonumber
=&2\sigma^4\theta_0^TP^{-T}
\Sigma^{-2}
\frac{\partial^2 P^{-1}}{\partial \eta_i\partial \eta_j}\theta_0
+2\sigma^4\theta_0^T\frac{\partial P^{-T}}{\partial \eta_j}
\Sigma^{-2}\frac{\partial P^{-1}}{\partial \eta_i} \theta_0\\
&-2\sigma^4{\rm Tr}\Big( \Sigma^{-1} \frac{\partial^2 P^{-1}}{\partial \eta_i\partial \eta_j} \Sigma^{-1}\Big).\label{hmseg}
\end{align}
\endgroup
The Hessian matrix $\frac{\partial^2 \overline{{\rm MSE}}{\rm
g}(P(\eta))}{\partial \eta \partial \eta^T}$ of $\overline{{\rm
MSE}}{\rm g}(P(\eta))$ is omitted here for simplicity. Then, it can
be shown that
\begin{align*}
\Big\|\frac{\partial^2 \overline{{\rm MSE}}{\rm g}(P(\eta))}{\partial \eta\partial \eta^T}
\!-\!\frac{\partial^2 W_g(P(\eta),\Sigma,\theta_0)}{\partial \eta\partial \eta^T}\Big\|
=o_p(1)
\end{align*}
uniformly for any $\eta \in\overline{\Omega}$. Applying the Taylor
expansion to $\frac{\partial \overline{{\rm MSE}}{\rm
g}(P(\eta))}{\partial \eta}$ yields
\begin{align*}
0=\frac{\partial \overline{{\rm MSE}}{\rm g}(P(\eta))}{\partial \eta}\Big|_{\eta = \widehat{\eta}_{\rm MSEg}}
=\frac{\partial \overline{{\rm MSE}}{\rm g}(P(\eta))}{\partial \eta}\Big|_{\eta = \eta_{\rm g}^*}\\
+\frac{\partial^2 \overline{{\rm MSE}}{\rm g}(P(\eta))}{\partial \eta \partial \eta^T}\Big|_{\eta = \bar{\eta}}
(\widehat{\eta}_{\rm MSEg} - \eta_{\rm g}^*),
\end{align*}
where $\bar{\eta}$ lies between $\widehat{\eta}_{\rm MSEg}$ and
$\eta_{\rm g}^*$.

Clearly, $$\frac{\partial^2 W_g(P(\eta),\Sigma,\theta_0)}{\partial
\eta\partial \eta^T}\Big|_{\eta = \eta_{\rm g}^*}=O_p(1).$$ Then
under Assumption \ref{ass:2}, we have $\frac{\partial^2
W_g(P(\eta),\Sigma,\theta_0)}{\partial \eta\partial
\eta^T}\Big|_{\eta = \eta_{\rm g}^*}$ is positive definite. For
sufficiently large $N$, $\bar{\eta}$ would be close to $\eta_{\rm
g}^*$. In this case, we also have $\frac{\partial^2
W_g(P(\eta),\Sigma,\theta_0)}{\partial \eta\partial
\eta^T}\Big|_{\eta = \bar \eta}$ is positive definite. Then it
follows that
\begin{align*}
&\widehat{\eta}_{\rm MSEg} - \eta_{\rm g}^*\\
 =& -\Big( \frac{\partial^2\overline{{\rm MSE}}{\rm g}(P(\eta))}{\partial \eta \partial \eta^T}\Big|_{\eta = \bar{\eta}} \Big)^{-1}
\frac{\partial \overline{{\rm MSE}}{\rm g}(P(\eta))}{\partial \eta}\Big|_{\eta = \eta_{\rm g}^*}\\
=&O_p(1)O_p(\varpi_N) =O_p(\varpi_N).
\end{align*}
Now, we prove $\|\widehat{\eta}_{\rm Sg} - \eta_{\rm
g}^*\|=O_p(\mu_N)$ and the proof is similar to that of
$\|\widehat{\eta}_{\rm MSEg} -\eta_{\rm g}^*\|=O_p(\varpi_N)$. By
\dref{ga5}, the $i$-th element of gradient vector of
$\mathscr{\overline{F}}_{\rm Sg}(P(\eta))$ is
\begin{align}
\nonumber
&\frac{\partial \mathscr{\overline{F}}_{\rm Sg}(P(\eta))}{\partial \eta_i}
\!=\!2\sigma^4
(\widehat{\theta}^{\rm LS})^T\!S^{-T}\!N^2(\Phi^T\Phi)^{-2}\frac{\partial S^{-1}}{\partial \eta_i}\widehat{\theta}^{\rm LS}\\
&\hspace{7em}+2\sigma^2N^2
{\rm Tr}
\Big(\frac{\partial R^{-1}}{\partial \eta_i}\Big). \label{gsg}
\end{align}
Using the identity $\frac{\partial R^{-1}}{\partial \eta_i}\! =\!-\sigma^2R^{-1}\frac{\partial P^{-1}}{\partial \eta_i}R^{-1}$, we see
\begingroup
\allowdisplaybreaks
\begin{align}
\nonumber
&\frac{\partial \mathscr{\overline{F}}_{\rm Sg}(P(\eta))}{\partial \eta_i}
-\frac{\partial W_g(P(\eta),\Sigma,\theta_0)}{\partial \eta_i}
=2\sigma^4 \Upsilon'_1+2\sigma^4 {\rm Tr}\big(\Upsilon'_2\big)\\
\nonumber
&\mbox{where}
~~\Upsilon'_1=(\widehat{\theta}^{\rm LS})^T\!S^{-T}\!N^2(\Phi^T\Phi)^{-2}\frac{\partial S^{-1}}{\partial \eta_i}\widehat{\theta}^{\rm LS}\\
&\hspace{6em}-\theta_0^TP^{-T} \Sigma^{-2}
\frac{\partial P^{-1}}{\partial \eta_i} \theta_0\\
\nonumber
&\hspace{3.4em}\Upsilon'_2=\Sigma^{-1} \frac{\partial P^{-1}}{\partial \eta_i}  \Sigma^{-1}
-NR^{-1}\frac{\partial P^{-1}}{\partial \eta_i}NR^{-1}.
\end{align}
\endgroup
Since $\Phi^T\Phi/N \xra{} \Sigma$ and $v(t)$ is a white noise, we
have $\|\widehat{\theta}^{\rm LS} - \theta_0\|=O_p(1/\sqrt{N})$.
Then noting that $\|N(\Phi^T\Phi)^{-1} -
\Sigma^{-1}\|=O_p(\delta_N)$,
    $\|S^{-1} - P^{-1}\| =O_p(1/N)$, $\big\|\frac{\partial S^{-1}}{\partial \eta_i}
    - \frac{\partial P^{-1}}{\partial \eta_i}\big\|=O_p(1/N) $,
    $\big\|NR^{-1} - \Sigma^{-1}  \big\|=O_p(\delta_N)$,
    and $ \|NR^{-1}\|=O_p(1)$, $\|\widehat{\theta}^{\rm LS}\|=O_p(1)$,
    and $d_1\leq \|P(\eta)\|\leq d_2 $ and $\|S(\eta)^{-1}\| < \|(P(\eta))^{-1}\|\leq 1/d_1 $ for
    $\eta\in\overline{\Omega}$, yields
\begin{align*}
&|\Upsilon'_1| \! = \!\max\big(O_p(1/\sqrt{N}),O_p(1/N), O_p(\delta_N)\big)\!=O_p(\mu_N),\\
&|{\rm Tr}\big(\Upsilon'_2\big)| \! = \!O_p(\delta_N),
\end{align*}
uniformly in $\overline{\Omega}$. It follows that
$$\Big\|\frac{\partial \mathscr{\overline{F}}_{\rm Sg}(P(\eta))}{\partial \eta}
-\frac{\partial W_g(P(\eta),\Sigma,\theta_0)}{\partial \eta}\Big\|=O_p(\mu_N)$$
uniformly for any $\eta \in\overline{\Omega}$.
This implies
\begin{align}
\frac{\partial \mathscr{\overline{F}}_{\rm Sg}(P(\eta))}{\partial \eta}\Big|_{\eta = \eta_{\rm g}^*} =O_p(\mu_N).
\end{align}
Similarly, one can obtain the Hessian matrix of
$\mathscr{\overline{F}}_{\rm Sg}(P(\eta))$ and can show that
\begin{align}
\Big\|\frac{\partial^2 \mathscr{\overline{F}}_{\rm Sg}(P(\eta))}{\partial \eta\partial \eta^T}
-\frac{\partial^2 W_g(P(\eta),\Sigma,\theta_0)}{\partial \eta\partial \eta^T}\Big\|
=o_p(1)\label{chsg}
\end{align}
uniformly for any $\eta \in\overline{\Omega}$. Applying the Taylor
expansion of $\frac{\partial \mathscr{\overline{F}}_{\rm
Sg}(P(\eta))}{\partial \eta}$ shows
\begin{align*}
0=\frac{\partial \mathscr{\overline{F}}_{\rm Sg}(P(\eta))}{\partial \eta}\Big|_{\eta = \widehat{\eta}_{\rm Sg}}
=\frac{\partial \mathscr{\overline{F}}_{\rm Sg}(P(\eta))}{\partial \eta}\Big|_{\eta = \eta_{\rm g}^*}\\
+\frac{\partial^2 \overline{{\rm MSE}}{\rm g}(P(\eta))}{\partial \eta \partial \eta^T}\Big|_{\eta = \widetilde{\eta}}
(\widehat{\eta}_{\rm MSEg} - \eta_{\rm g}^*),
\end{align*}
where $\widetilde{\eta}$ lies between $\widehat{\eta}_{\rm Sg}$ and
$\eta_{\rm g}^*$. For sufficiently large $N$, we have
\begin{align*}
&\widehat{\eta}_{\rm Sg} - \eta_{\rm g}^*\\
=& -\Big( \frac{\partial^2 \mathscr{\overline{F}}_{\rm Sg}(P(\eta))}{\partial \eta \partial \eta^T}\Big|_{\eta = \widetilde{\eta}} \Big)^{-1}
\frac{\partial \mathscr{\overline{F}}_{\rm Sg}(P(\eta))}{\partial \eta}\Big|_{\eta = \eta_{\rm g}^*}\\
=&O_p(1)O_p(\mu_N)
=O_p(\mu_N).
\end{align*}
The proof of \dref{ohpsyr} and \dref{ohpmlr} can be done in a
similar way and thus is omitted. This completes the proof.

\section*{Appendix B}

\renewcommand{\thesection}{B}

\setcounter{subsection}{0}

\setcounter{lem}{0}
\renewcommand{\thelem}{B\arabic{lem}}
\setcounter{rem}{0}
\renewcommand{\therem}{B\arabic{rem}}

This appendix contains the technical lemmas used in the proof in
Appendix A.

\subsection{Matrix Differentials and Related Identities}

This section introduces the differentiation of a function $f(X)$
where $X$ is a matrix. It is assumed that $X$ has no special
structure, i.e., that the elements of $X$ are independent. For
convenience and readability, the formulas used in the paper are
stated in the following lemmas.

\begin{lem} \citep{Petersen2012}
    \label{md}
    Assume that $b$ is a column vector, and $A,B$ and $X$ are matrices with compatible dimensions.
    Then we have
    \begingroup
    \allowdisplaybreaks
    \begin{align}
    &\frac{\partial b^T X^T A X b}{\partial X}=(A+A^T)Xb b^T \label{md1}\\
    &\frac{\partial b^T X^{-1} b}{\partial X}=-X^{-T}b b^TX^{-T} \label{md4}\\
    &\frac{\partial \log |\det (X)|}{\partial X}=X^{-T} \label{md2}\\
    &\frac{\partial (X^{-1})_{kl}}{\partial X_{ij}}=-(X^{-1})_{ki} (X^{-1})_{jl} \label{md3}\\
    &\frac{\partial {\rm Tr}(AX^{-1}B) }{\partial X}=-(X^{-1}BAX^{-1})^{T} \label{md5}\\
    &\frac{\partial {\rm Tr}(AX B  X^T  A^T) }{\partial X}=A^TAX(B+B^T). \label{md6}
    \end{align} where $(\cdot)_{ij}$ denotes the $(i,j)$th element of a matrix.
\endgroup
\end{lem}

\begin{lem}\label{lm0}
    Suppose that both $A$ and $B$ are positive  semidefinite.
    If ${\rm Tr}(AB)= 0$, then $AB=0$.
\end{lem}
\begin{proof}
    Let us denote the symmetric square root factorization of $A$ by $A^{\frac12}$.
    Thus the trace property implies
    \begin{align*}
    {\rm Tr}(AB)=&{\rm Tr}(A^{\frac12}A^{\frac12}B^{\frac12}B^{\frac12})\\
    =&{\rm Tr}(A^{\frac12}B^{\frac12}B^{\frac12}A^{\frac12})
    =\|A^{\frac12}B^{\frac12}\|^2=0.
    \end{align*}
    This derives that $A^{\frac12}B^{\frac12}=0$.
    Pre-multiplying by $A^{\frac12}$ and post-multiplying by $B^{\frac12}$
    entails $AB=0$.
\end{proof}

\begin{lem}
We have the following identities:
\begingroup
\allowdisplaybreaks
\begin{align}
&\sum\nolimits_{ij}(A)_{ij}J_{ij} = A \label{id1},\\
& Y - \Phi \widehat{\theta}^{\rm R} = \sigma^2Q^{-1}Y\label{id3},\\
& \widehat{\theta}^{\rm LS} - \widehat{\theta}^{\rm R}
=\sigma^2(\Phi^T\Phi)^{-1}\Phi^TQ^{-1}Y\label{id4},\\
& A(I_N+BA)^{-1}=(I_n+AB)^{-1}A \label{id2},\\
&\Phi^TQ^{-1} \Phi =S^{-1},~~
\Phi^TQ^{-1}Y = S^{-1}\widehat{\theta}^{\rm LS}, \label{id6}\\
&\Phi^TQ^{-T}Q^{-1} \Phi =S^{-T} (\Phi^T\Phi)^{-1}S^{-1},\\
&\Phi^TQ^{-T}Q^{-1}Y = S^{-T} (\Phi^T\Phi)^{-1}S^{-1}\widehat{\theta}^{\rm LS}, \label{id5}\\
&I_N \!- \!\sigma^2 Q^{-1} \!=\! \Phi P \Phi^T Q^{-1} \!=\! Q^{-1}\Phi P \Phi^T\!\!=\!\Phi R^{-1}\Phi^T, \label{id9}\\
&\frac{\partial (Q^{-1})_{ij}}{\partial P} = -\Phi^T Q^{-T} J_{ij} Q^{-T}\Phi \label{id7},\\
&\frac{\partial (R^{-1})_{ij}}{\partial P} = \sigma^2P^{-T} R^{-T}
J_{ij} R^{-T}P^{-T}, \label{id10}
\end{align} \endgroup
where $J_{ij}$ is a matrix whose $(i,j)$-element is one and zero for all other elements.
\end{lem}
\begin{proof}
    The identities \dref{id1}--\dref{id9} can be verified by a straightforward calculation.
    Using \dref{md3} gives
    \begingroup
    \allowdisplaybreaks
    \begin{align*}
    \frac{\partial (Q^{-1})_{ij}}{\partial P_{st}}
    =&\sum_{a,b}    \frac{\partial (Q^{-1})_{ij}}{\partial Q_{ab}}
    \frac{\partial Q_{ab}}{\partial P_{st}}\\
    =&-\sum_{a,b}(Q^{-1})_{ia}(Q^{-1})_{bj}\Phi_{as}(\Phi^T)_{tb}\\
    =&-\sum_{a,b}(\Phi^T)_{sa}(Q^{-T})_{ai}(Q^{-T})_{jb}\Phi_{bt}\\
    =&-(\Phi^T Q^{-T})_{si}(Q^{-T}\Phi )_{jt},
    \end{align*}
    \endgroup
    which implies \dref{id7}.
    While \dref{id10} can be proved in a similar way.
\end{proof}

\subsection{Convergence Result for Extremum Estimators}

\begin{lem}\citep[Theorem 8.2]{Ljung1999}
    \label{ct}
Assume that
\begin{enumerate}[1)]
    \item $M(\eta)$ is a deterministic function that is continuous in $\eta \in \Omega$ and  minimized at the set
    \begin{align*}
    D\!=\!\argmin_{\eta \in \Omega} M(\eta) \!=\! \big\{\eta|\eta \in
    \Omega, M(\eta)\!=\!\min_{\eta'\in \Omega }M(\eta')\big\}
    \end{align*}
    where $\Omega$ is a compact subset of $\mathbb{R}^p$.

    \item A sequence of functions $\{M_N(\eta)\}$ converges to $M(\eta)$ almost surely and uniformly in $\Omega$ as $N$ goes to $\infty$.

\end{enumerate}
    Then $\widehat{\eta}_N = \arg\min_{\eta\in \Omega} M_N(\eta)$
    converges to $D$ almost surely, namely,
    \begin{align*}
    \inf_{\bar{\eta}\in D}\|\widehat{\eta}_N-\bar{\eta}\|\xra{}0,~~\mbox{as}~N\xra{}\infty.
    \end{align*}
\end{lem}



\end{document}